\documentclass[sigplan,screen]{acmart}

\usepackage{booktabs}
\usepackage{subcaption}
\usepackage{microtype}
\usepackage{mathtools}
\usepackage{paralist}
\usepackage{extarrows} 
\usepackage{tikz}
\usetikzlibrary{arrows,automata,shapes.misc,shapes,snakes}
\usepackage{wrapfig}

\newcommand{\tsup}{\textstyle\sup}
\newcommand{\lt}{\langle}
\newcommand{\rt}{\rangle}
\newcommand{\dom}{\mathop{\mathrm{dom}}}
\newcommand{\kstar}{^{\textstyle *}}
\newcommand{\N}{\mathbb{N}}

\newcommand{\sem}[1]{\llbracket #1 \rrbracket}
\newcommand{\eps}{\varepsilon}
\renewcommand{\phi}{\varphi}
\renewcommand{\theta}{\vartheta}
\newcommand{\caret}{\text{\textasciicircum}}

\newcommand{\CReg}{\mathit{CReg}}
\newcommand{\aut}[1]{\mathcal{#1}}
\newcommand{\Tokens}{\mathbf{Tk}}
\newcommand{\Cfg}{\mathbf{C}}
\newcommand{\degree}{\mathsf{degree}}
\newcommand{\Slash}{\mathbin{/}}
\newcommand{\pp}{{+}{+}}

\copyrightyear{2022}
\acmYear{2022}
\setcopyright{acmlicensed}
\acmConference[PLDI '22]{Proceedings of the 43rd ACM SIGPLAN International Conference on Programming Language Design and Implementation}{June 13--17, 2022}{San Diego, CA, USA}
\acmBooktitle{Proceedings of the 43rd ACM SIGPLAN International Conference on Programming Language Design and Implementation (PLDI '22), June 13--17, 2022, San Diego, CA, USA}
\acmPrice{15.00}
\acmDOI{10.1145/3519939.3523456}
\acmISBN{978-1-4503-9265-5/22/06}

\begin{document}

\title{Software-Hardware Codesign for Efficient In-Memory Regular Pattern Matching}

\author{Lingkun Kong}
\affiliation{
  \institution{Rice University}            
  \country{USA}                    
}
\email{klk@rice.edu}
\authornote{These authors contributed equally.}

\author{Qixuan Yu}
\affiliation{
  \institution{Rice University}            
  \country{USA}                    
}
\email{qy12@rice.edu}
\authornotemark[1]

\author{Agnishom Chattopadhyay}
\affiliation{
  \institution{Rice University}            
  \country{USA}                    
}
\email{agnishom@rice.edu}

\author{Alexis Le Glaunec}
\affiliation{
  \institution{Rice University}            
  \country{USA}                    
}
\email{alexis.leglaunec@rice.edu}

\author{Yi Huang}
\affiliation{
  \institution{Rice University}            
  \country{USA}                    
}
\email{781013488@qq.com}          

\author{Konstantinos Mamouras}
\affiliation{
  \institution{Rice University}            
  \country{USA}                    
}
\email{mamouras@rice.edu}

\author{Kaiyuan Yang}
\affiliation{
  \institution{Rice University}            
  \country{USA}                    
}
\email{kyang@rice.edu}

\begin{abstract}
Regular pattern matching is used in numerous application domains, including text processing, bioinformatics, and network security. Patterns are typically expressed with an extended syntax of regular expressions. This syntax includes the computationally challenging construct of bounded repetition or counting, which describes the repetition of a pattern a fixed number of times. We develop a specialized in-memory hardware architecture that integrates counter and bit vector modules into a state-of-the-art in-memory NFA accelerator. The design is inspired by the theoretical model of nondeterministic counter automata (NCA). A key feature of our approach is that we statically analyze regular expressions to determine bounds on the amount of memory needed for the occurrences of bounded repetition. The results of this analysis are used by a regex-to-hardware compiler in order to make an appropriate selection of counter or bit vector modules. We evaluate our hardware implementation using a simulator based on circuit parameters collected by SPICE simulation in TSMC 28nm CMOS process. 
We find that the use of counter and bit vector modules outperforms unfolding solutions by orders of magnitude. Experiments concerning realistic workloads show up to 76\% energy reduction and 58\% area reduction in comparison to CAMA, a recently proposed in-memory NFA accelerator.
\end{abstract}



\maketitle

\renewcommand{\shortauthors}{L. Kong, Q. Yu, A. Chattopadhyay, A. Le Glaunec, Y. Huang, K. Mamouras, and K. Yang}

\section{Introduction}
\label{sec:intro}

Regular pattern matching, where the patterns are expressed with finite-state automata or regular expressions, has numerous applications in text search and analysis \cite{AhoAC75}, network security \cite{YuCDLK2006}, bioinformatics \cite{RoyA2016, BoDSS2018}, and runtime verification \cite{BarringerGHS2004, BartocciDDFMNS2018}. 
Various techniques have been developed for matching regular patterns, many of which are based on the execution of deterministic finite automata (DFAs) or nondeterministic finite automata (NFAs). DFA-based techniques are generally faster, as the processing of an input element requires a single memory lookup, while NFA-based techniques are slower, as they involve extending several execution paths when processing one element. The advantage of NFAs over DFAs is that they are typically more memory-efficient, and there are cases where an equivalent DFA would unavoidably be exponentially larger \cite{MeyerF1971}.

Many applications require the processing of large and complex NFAs on real-time streams of data collected from sensors, networks, and various system traces. Energy efficiency and memory efficiency (in terms of the memory capacity or chip footprint needed for a given NFA) are highly desirable for both high-performance computing and battery-powered embedded applications.
NFA processing requires frequent, yet irregular and unpredictable, memory accesses on general-purpose processors, leading to limited throughput and high power on CPU and GPU architectures~\cite{ANMLZoo, LenjaniMH14, LiuTY11}. Field Programmable Gate Arrays (FPGAs) offer high speed through hardware-level parallelism, but are often bottlenecked by routing congestion~\cite{RahimiRS20, XieTV17} and their high power, area and cost prevent their use in mobile and embedded devices. Even with digital application-specific integrated circuit (ASIC) accelerators, the memory access bandwidth restricts the parallelism~\cite{tandonPS16, VanLJ2012}. 
The latest hardware technology that addresses these challenges is in-memory architecture, which processes the NFA transitions directly inside memories with massive parallelism and merged memory and computing operations. For instance, the Automata Processor (AP) from Micron~\cite{DlugoschBGLN2014AP, WangKA16} outperforms x86 CPUs by 256×, GPGPUs by 32×, and the digital accelerator XeonPhi by 62× in the ANMLZoo benchmark suite \cite{ANMLZoo, SubramaniyanAW17}.

Classical regular expressions (regexes) involve operators for concatenation $\cdot$, nondeterministic choice $+$, and iteration (Kleene's star) $\kstar$. They can be translated into NFAs whose size is linear in the size of the regex \cite{Thompson1968,Glushkov1961Abstract}. However, the regexes used in practice have several additional features that make them more succinct. One such feature is \emph{counting}, written as $r\{m,n\}$, which is also called \emph{constrained} or \emph{bounded repetition}. The pattern $r\{m,n\}$ expresses that the subpattern $r$ is repeated anywhere from $m$ to $n$ times. This counting operator is ubiquitous in practical use cases of regexes. For example, we have observed that in several datasets for network intrusion detection (Snort \cite{Snort} and Suricata \cite{Suricata}) and motif search in biological sequences (Protomata \cite{Prosite, RoyA2016}) counting arises in the majority of the patterns. The naive approach for dealing with counting operators is to rewrite them by unfolding. For example, $r\{n,n\}$ is unfolded into $r \cdot r \cdots r$ ($n$-fold concatenation) and results in an NFA of size linear in $n$ (and therefore can produce a DFA of size exponential in $n$).
Since $n$ can grow very large, dealing with counting is one of the main technical challenges for successfully using hardware-based approaches to execute practical regular patterns.

Existing in-memory NFA architectures use this naive unfolding method to handle counting operators. This leads to the use of a large number of STEs\footnote{STE stands for State Transition Element \cite{DlugoschBGLN2014AP}. It is a hardware element that roughly corresponds to the state of a homogeneous NFA. It contains a state bit (to indicate whether the state is active or not) and a memory array that represents a character class.
} to support counting. In AP~\cite{DlugoschBGLN2014AP} and CA (Cache Automaton)~\cite{SubramaniyanAW17}, each STE uses 256 memory bits for 8-bit symbols. In the latest Impala~\cite{SadrediniER20} and CAMA\footnote{CAMA abbreviates Content Addressable Memory (CAM) enabled Automata accelerator.}~\cite{cama} designs, each STE requires 16 to 32 memory bits. Even with this improvement, a modest counting operator with upper limit 1024 requires at least $16384$ memory bits, while the information required for implementing the operator may be only $10$ bits in some cases. Unfolding counting operators results in large memory and energy usage. To circumvent these problems, we explore software and hardware co-design for integrating counter and bit vector modules into a state-of-the-art in-memory NFA architecture.

Our design is inspired by an extension of NFAs with counter registers called nondeterministic counter automata (NCAs). In an NCA, a computation path involves not only transitions between control states, but also the use of a finite number of registers that hold nonnegative integers. Such automata are a natural execution model for regexes with counting, as the counters can track the number of repetitions of subpatterns. When the counters are bounded, NCAs are expressively equivalent to NFAs, but they can be exponentially more succinct \cite{MeyerF1971, StockmeyerM1973}. Similar to how an NFA is executed by maintaining the set of active states, an NCA is executed by maintaining a set of pairs, which we call \emph{tokens}, where the first component is the control state and the second component specifies the values of the counters. A key idea of our approach is that we can statically analyze an NCA to determine which states can carry a large number of tokens during execution. We call a control state \textbf{\em counter-unambiguous} if it can only carry at most one token and \emph{counter-ambiguous} if it can carry more than one. 
In the case of counter-unambiguity for a state $q$ with counter $x$, we know that we only need to record one counter value, which means that we need only one memory location whose size (in bits) is logarithmic in the range $M$ of possible counter values. 
In the case of counter-ambiguity for $q$ with counter $x$, we may have to record a large number of counter values (as large as $M$), and our insight is to use a bit vector $v$ of size $M$, where $v[i] = 1$ (resp., $v[i] = 0$) indicates the presence (resp., absence) of a token at $q$ with counter value $i$. So, identifying a state as counter-unambiguous enables a massive memory reduction for this state from $O(M)$ to $O(\log M)$.

We design a \textbf{\em static analysis algorithm} for checking the counter-ambiguity of NCAs and regexes by performing a systematic exploration of the space of reachable tokens to identify the existence of some input string for which two different tokens are placed on the same control state. This may lead to a large search space (exponential in the size of the regex), and the worst case is not easy to avoid since the problem is NP-hard. To handle difficult instances that involve large repetition bounds, we also provide an \emph{over-approximate} algorithm that gives an inconclusive output for some instances, while still being able to identify cases of counter-unambiguity for most instances from real benchmarks. By combining the exact and over-approximate algorithms, we can statically analyze within milliseconds the vast majority of regexes
in the benchmarks Snort \cite{Snort}, Suricata \cite{Suricata}, Protomata \cite{RoyA2016}, SpamAssassin \cite{SpamAssassin}, and ClamAV \cite{ClamAV}.

Using the insights about NCA execution mentioned earlier, we propose a \textbf{\em hardware design} that is based on existing in-memory NFA architectures (AP, CA, Impala, CAMA) augmented with (1) \emph{counter modules} for counter-unambiguous states, and (2) \emph{bit vector modules} for counter-ambiguous states. We use SPICE \cite{SPICE}, an industry-standard simulator for integrated circuits, to perform hardware simulation for the counters and bit vectors and to integrate them into the CAMA architecture.
We also provide a \textbf{\em compiler} that statically analyzes an input regex to determine counter-(un)ambiguity and then creates a representation of an automaton with counters and bit vectors using the MNRL format \cite{mnrl} that can be used to program the hardware. 
Several existing architectures like AP provide a counter module in their design, but they typically do not provide a compiler that translates regexes to hardware-recognizable programs.
Also, counter registers alone cannot deal with the challenging instances of counting. 
Compared with prior works that do not provide a bit vector module, this paper proposes a novel design that can systematically handle counting and ensure correct compilation in both the easy (requiring counters) and difficult (requiring bit vectors) cases.

We modified the open-source simulator VASim \cite{ANMLZoo} to simulate the hardware performance of our counter- and bit-vector-augmented CAMA design with implementation in TSMC 28nm process. In microbenchmarks, we evaluated the energy and area consumption of counters and bit vectors against their unfolded counterparts. The results show that our counter- and bit-vector-based design can reduce the energy usage by orders of magnitude and the area by large margins.
Furthermore, we evaluated the performance of the augmented CAMA design using the Snort \cite{Snort}, Suricata \cite{Suricata}, Protomata \cite{RoyA2016}, and SpamAssassin \cite{SpamAssassin} benchmarks. For applications involving regexes with large counting bounds, the results show as large as 76\% energy reduction and 58\% area reduction. For regexes with small counting bounds, the results show little to no overhead.

\paragraph{Contributions.}

The main contributions of this paper are summarized below:
\begin{asparaenum}[(1)]
\item
We use the notion of \emph{counter-unambiguity} in order to identify instances of bounded repetition that can be handled with a small amount of memory. We describe both an exact and an over-approximate \emph{static analysis} for counter-(un)ambiguity which, when combined, allow us to efficiently analyze the regexes that arise in several application domains.
\item
We propose a \emph{hardware design} that augments the prior NFA-based CAMA architecture \cite{cama} with counter and bit vector modules, which are inspired from the execution of NCAs and the classification of states as counter-(un)ambiguous. This architecture achieves substantial energy and area reductions compared to prior designs.
\item
We provide a \emph{compiler} that enables the high-level programming of the hardware using POSIX-style regexes. The compiler first performs the static analysis for counter-(un)ambiguity and then leverages the analysis results for producing a low-level description of the automaton.
\end{asparaenum}

\section{Preliminaries}
\label{sec:prelim}

In this section, we will give a brief overview of several well-known concepts, including regular expressions with counting and nondeterministic counter automata (NCAs). We are not interested in NCAs with unbounded counters (which can recognize non-regular languages), so we focus on NCAs with bounded counters. These automata are an appropriate model for implementing regular expressions with counting. Differently from most definitions of NCAs in the literature, we allow each control state of the automaton to have a different number of counters. This flexibility allows us to carefully bound the memory needed for NCA execution.

Let $\Sigma$ be a finite alphabet. A \textbf{\em regular expression} (or \emph{regex}) over $\Sigma$ is given by the grammar
$r ::=
 \eps \mid
 \sigma \mid
 r \cdot r \mid
 r + r \mid
 r\kstar \mid
 r\{m,n\} 
$,
where $\sigma \subseteq \Sigma$ is a predicate over the alphabet and $m,n$ are natural numbers. The expression $r\{m,n\}$ describes the repetition of $r$ from $m$ to $n$ times, so we require that $0 \leq m \leq n$. We write $r\{n\}$ for $r\{n, n\}$. The concatenation symbol is sometimes omitted, i.e., we write $r_1 r_2$ instead of $r_1 \cdot r_2$.
The \emph{interpretation} of a regex $r$ is a language $\sem{r} \subseteq \Sigma \kstar$, which is defined in the standard way.

\textbf{\em Notation for predicates}:
A predicate over the alphabet is sometimes referred to as a \emph{character class}. The predicate $\Sigma$ contains all symbols in the alphabet.
When we use a symbol $a \in \Sigma$ in a regex, it should be understood as the singleton predicate $\{ a \} \subseteq \Sigma$. We will also use the notation $[a_1 \ldots a_n]$ in a regex to represent the predicate $\{ a_1, \ldots, a_n \} \subseteq \Sigma$. We write $[\caret a_1 \ldots a_n]$ for the predicate $\Sigma \setminus \{ a_1, \ldots, a_n \}$ that contains all symbols aside from $a_1, \ldots, a_n$. For a predicate $\sigma \subseteq \Sigma$, we write $\bar\sigma = \Sigma \setminus \sigma$ to denote its \emph{complement}.


We fix an infinite set $\CReg$ of counter registers or, simply, \emph{counters}. We typically write $x, y, z, \ldots$ to denote counter registers. For a subset $V \subseteq \CReg$ of counters, we say that a function $\beta: V \to \N$, which assigns a value to each counter in $V$, is a \emph{$V$-valuation}.

\begin{definition}
\label{def:NCA}
Let $\Sigma$ be a finite alphabet. A \emph{nondeterministic counter automaton (NCA)} with input alphabet $\Sigma$ is a tuple $\aut A = (Q,R,\Delta,I,F)$, where
\begin{itemize}[$-$]
\item
$Q$ is a finite set of \emph{states},
\item
$R: Q \to \mathcal{P}(\CReg)$ is a function that maps each state to a finite set of counters,
\item
$\Delta$ is the \emph{transition relation}, which contains finitely many transitions of the form $(p, \sigma, \phi, q, \theta)$, where $p$ is the source state, $\sigma \subseteq \Sigma$ is a predicate over the alphabet, $\phi \subseteq (R(p) \to \N)$ is a predicate over $R(p)$-valuations, $q$ is the destination state, and $\theta: (R(p) \to \N) \to (R(q) \to \N)$,
\item
$I$ is the \emph{initialization function}, a partial function defined on the subset $\dom(I) \subseteq Q$ of \emph{initial states} that specifies an \emph{initial valuation} $I(q): R(q) \to \N$ for each initial state $q$, and
\item
$F$ is the \emph{finalization function}, a partial function defined on the subset $\dom(F) \subseteq Q$ of \emph{final states} that specifies a predicate $F(q) \subseteq R(q) \to \N$ for each final state $q$.
\end{itemize}
We say that a state $q \in Q$ is \emph{pure} if $R(q) = \emptyset$, that is, it has no counter associated with it.
\end{definition}

We remark that the states in an NCA of Definition~\ref{def:NCA} do not necessarily have the same counters. In fact, some states may not have any counter at all. In a transition $(p, \sigma, \phi, q, \theta)$, we will call the predicate $\phi$ a \emph{guard} because it may restrict a transition based on the values of the counters, and we will call the function $\theta$ an \emph{action}, because it describes how to assign counter values in the destination state given the counter values in the source state.

We convert regexes (with counting) to NCAs that recognize the same language using a variant of the Glushkov construction \cite{Glushkov1961Abstract, GeladeGM09}. 
In contrast to Thompson's construction \cite{Thompson1968}, Glushkov's construction results in $\eps$-free automata that are also \emph{homogeneous}, i.e., all incoming transitions of a state are labeled with the same predicate over the alphabet. 
We present below several examples of NCAs.

\begin{figure*}
\centering
\begin{tikzpicture}[->, >=to, auto, node distance=1cm, semithick, scale=0.9, transform shape]
\small
\node (Pre) {};
\node[draw, rounded rectangle, right of=Pre] (A) {$q_1$};
\node[draw, rounded rectangle, right of=A, node distance=1.5cm] (B) {$q_2$};
\node[draw, rounded rectangle, right of=B, node distance=3cm] (C) {$q_3: x$};
\node[draw, rounded rectangle, right of=C, node distance=3cm] (D) {$q_4: x,y$};
\node[draw, rounded rectangle, right of=D, node distance=2.5cm] (E) {$q_5: x,y$};
\node[draw, rounded rectangle, right of=E, node distance=3.5cm] (F) {$q_6: x$};
\node[draw, rounded rectangle, right of=F, node distance=2.5cm, accepting] (G) {$q_7$};
\node[right of=G] (Post) {};
\path (Pre) edge (A);
\path (A) edge[loop below] node[right, xshift=0.5ex] {$\Sigma$} (A);
\path (A) edge node {$\sigma_1$} (B);
\path (B) edge node {$\sigma_2 \Slash x \coloneqq 1$} (C);
\path (C) edge node {$\sigma_3 \Slash y \coloneqq 1$} (D);
\path (D) edge[bend left=10] node {$\sigma_4$} (E);
\path (E) edge[bend left=10] node {$\sigma_3, y < n \Slash y\pp$} (D);
\path (E) edge node {$\sigma_5, m \leq y \leq n$} (F);
\path (F) edge[bend left=17] node[pos=0.15, yshift=1ex] {$\sigma_2, x<k \Slash x\pp$} (C);
\path (F) edge node {$\sigma_6, x = k$} (G);
\path (G) edge (Post);
\end{tikzpicture}
\caption{NCA with two counters ($x$ and $y$) for the regex $\Sigma\kstar \sigma_1 (\sigma_2 (\sigma_3 \sigma_4)\{m,n\} \sigma_5)\{k\} \sigma_6$ with $1 \leq m \leq n$ and $k \geq 1$.}
\label{fig:NCA_two_counters}
\end{figure*}
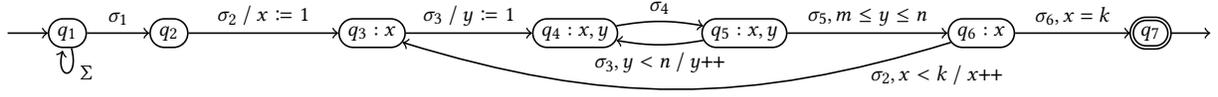

\begin{example}
\label{ex:NCAs}
Consider the regex $r_1 = \Sigma\kstar \sigma_1 \sigma_2 \{n\}$ with $n \geq 1$, where $\sigma_1, \sigma_2$ are predicates over the alphabet\footnote{In order to make the example more concrete, suppose that $\sigma_1 = [a b]$ and $\sigma_2 = [\caret a]$. So, the regular expression $r_1$ is the same as $\mathtt{.\kstar[ab][\caret a]\{n\}}$ using POSIX notation \cite{PosixSyntax}. Note that $\Sigma\kstar$ is the same as $\mathtt{.}\kstar$ in POSIX notation.}. The following automaton recognizes the language of $r_1$:
\[
\begin{tikzpicture}[->, >=to, auto, node distance=1.75cm, semithick, scale=0.9, transform shape]
\small
\node (Pre) {};
\node[draw, rounded rectangle, right of=Pre, node distance=1cm] (A) {$q_1$};
\node[draw, rounded rectangle, right of=A] (B) {$q_2$};
\node[draw, rounded rectangle, right of=B, node distance=3cm, accepting] (C) {$q_3: x$};
\node[right of=C] (Post) {};
\path (Pre) edge (A);
\path (A) edge[loop below] node[right, xshift=0.5ex, yshift=1ex] {$\Sigma$} (A);
\path (A) edge node {$\sigma_1$} (B);
\path (B) edge node {$\sigma_2 \Slash x \coloneqq 1$} (C);
\path (C) edge[loop below] node[right, xshift=0.5ex, yshift=1ex] {$\sigma_2, x<n \Slash x\pp$} (C);
\path (C) edge node {$x = n$} (Post);
\end{tikzpicture}
\]
The automaton above has three states: $q_1$, $q_2$, and $q_3$. We write $q_3: x$ to indicate that $R(q_3) = \{ x \}$. Notice that $q_1$ has no annotation with counters, which means that $R(q_1) = \emptyset$ (i.e., $q_1$ is pure). We annotate each edge $p \to q$ with an expression of the form $\sigma, \phi \Slash \theta$, where $\sigma$ is a predicate over $\Sigma$, $\phi$ is a guard over the counters of $p$, and $\theta$ is an assignment for the counters of $q$ using the counters of $p$. If the guard $\phi$ is omitted, then it is always true. The action $\theta$ is omitted only when $R(q) \subseteq R(p)$, and the omission indicates that the counters $R(q)$ retain the values from the previous state. We can also indicate this explicitly by writing ``$x \coloneqq x$''.
We write ``$x = n$'' for the guard that checks whether the value of counter $x$ is equal to $n$, and we write ``$x \coloneqq n$'' to denote the assignment (action) of the value $n$ to the counter $x$.
We use double circle notation to indicate that a state is final (see state $q_3$ above). An arrow emanating from a final state $q$ is annotated with the predicate $F(q)$ over counter valuations (recall that $F$ is the finalization function).

The regex $r_2 = \Sigma\kstar \sigma_1 (\sigma_2 \sigma_3)\{m,n\} \sigma_4$ with $1 \leq m \leq n$ is recognized by the following automaton:
\[
\begin{tikzpicture}[->, >=to, auto, node distance=0.75cm, semithick, scale=0.9, transform shape]
\footnotesize
\node (Pre) {};
\node[draw, rounded rectangle, right of=Pre] (A) {$q_1$};
\node[draw, rounded rectangle, right of=A, node distance=1.1cm] (B) {$q_2$};
\node[draw, rounded rectangle, right of=B, node distance=2cm] (C) {$q_3: x$};
\node[draw, rounded rectangle, right of=C, node distance=1.75cm] (D) {$q_4: x$};
\node[draw, rounded rectangle, right of=D, node distance=2.75cm, accepting] (E) {$q_5$};
\node[right of=E] (Post) {}; %
\path (Pre) edge (A);
\path (A) edge[loop below] node[right, xshift=0.5ex, yshift=1ex] {$\Sigma$} (A);
\path (A) edge node {$\sigma_1$} (B);
\path (B) edge node {$\sigma_2 \Slash x \coloneqq 1$} (C);
\path (C) edge[bend right=17] node[swap] {$\sigma_3$} (D);
\path (D) edge[bend right=17] node[swap] {$\sigma_2, x < n \Slash x\pp$} (C);
\path (D) edge node {$\sigma_4, m \leq x \leq n$} (E);
\path (E) edge (Post);
\end{tikzpicture}
\]
The regex $r_3 = \sigma_1\{m\} \Sigma\kstar \sigma_2\{n\}$ with $m, n \geq 1$ is recognized by the following automaton:
\[
\begin{tikzpicture}[->, >=to, auto, node distance=1.5cm, semithick, scale=0.9, transform shape]
\footnotesize
\node (Pre) {};
\node[draw, rounded rectangle, right of=Pre, node distance=1cm] (A) {$q_1$};
\node[draw, rounded rectangle, right of=A, node distance=2cm] (B) {$q_2$};
\node[draw, rounded rectangle, right of=B, node distance=2cm] (C) {$q_3: x$};
\node[draw, rounded rectangle, right of=C, node distance=2.5cm, accepting] (D) {$q_4: x$};
\node[right of=D] (Post) {};
\path (Pre) edge (A);
\path (A) edge node {$\sigma_1 \Slash x \coloneqq 1$} (B);
\path (B) edge[loop below] node {$\sigma_1, x < m \Slash x\pp$} (B);
\path (B) edge node {$\Sigma, x = m$} (C);
\path (B) edge[bend right] node[below] {$\sigma_2, x = m \Slash x \coloneqq 1$} (D);
\path (C) edge[loop below] node[right, yshift=1ex] {$\Sigma$} (C);
\path (C) edge node {$\sigma_2 \Slash x \coloneqq 1$} (D);
\path (D) edge[loop below] node {$\sigma_2, x<n \Slash x\pp$} (D);
\path (D) edge node {$x = n$} (Post);
\end{tikzpicture}
\]
All automata so far use one counter. For the regex $r_4 = \Sigma\kstar \sigma_1 (\sigma_2 (\sigma_3 \sigma_4)\{m,n\} \sigma_5)\{k\} \sigma_6$ with $1 \leq m \leq n$ and $k \geq 1$ we need two counters. See Fig.~\ref{fig:NCA_two_counters}.
\end{example}

\paragraph{Nondeterministic semantics}

Let $\aut A$ be an NCA. A \emph{token} for $\aut A$ is a pair $(q,\beta)$, where $q$ is a state and $\beta: R(q) \to \N$ is a counter valuation for $q$. The set of all tokens for $\aut A$ is denoted by $\Tokens(\aut A)$. For a letter $a \in \Sigma$, we define the \emph{token transition relation} $\to^a$ on $\Tokens(\aut A)$ as follows: $(p,\beta) \to^a (q,\gamma)$ if there is a transition $(p,\sigma,\phi,q,\theta) \in \Delta$ with $a \in \sigma$ such that $\beta \in \phi$ and $\gamma = \theta(\beta)$.
A token $(q,\beta)$ is \emph{initial} if the state $q$ is initial. A token $(q,\beta)$ is \emph{final} if the state $q$ is final and $\beta \in F(q)$. A \emph{run} of $\aut A$ on a string $a_1 a_2 \ldots a_n \in \Sigma\kstar$ is a sequence
\[
  (q_0,\beta_0) \xlongrightarrow{a_1}
  (q_1,\beta_1) \xlongrightarrow{a_2}
  (q_2,\beta_2) \xlongrightarrow{a_3}
  \cdots \xlongrightarrow{a_n}
  (q_n,\beta_n),
\]
where each $(q_i,\beta_i)$ is a token, $q_0$ is an initial state and $\beta_0 = I(q_0)$, and $(q_{i-1},\beta_{i-1}) \to^a (q_i,\beta_i)$ for every $i = 1,\ldots,n$. A run is \emph{accepting} if it ends with a final token. The NCA $\aut A$ \emph{accepts} a string if there is an accepting run on it. We write $\sem{\aut A} \subseteq \Sigma\kstar$ for the set of strings that $\aut A$ accepts.

Notice that, for a NCA $\aut A$, the set of tokens $\Tokens(\aut A)$ together with the transition relations $\to^a$ forms a labeled transition system. The family of transition relations $(\to^a)_{a \in \Sigma}$ can be represented as a ternary relation ${\to} \subseteq \Tokens(\aut A) \times \Sigma \times \Tokens(\aut A)$.

\textbf{\em Notation for tokens}: For a pure state $q$ (i.e., a state with no counter, see Definition~\ref{def:NCA}), there is only one valuation, denoted $0_\N: \emptyset \to \N$, which carries no information. So, we will often abuse notation and simply write $q$ for the token $(q,0_\N)$. Similarly, for a state $q$ with one counter, i.e., $R(q) = \{ x \}$ for some $x \in \CReg$, a valuation $\beta$ (of type $\{ x \} \to \N$) for $q$ specifies only one value $c = \beta(x)$ for the unique variable $x$ for $q$. For this reason, we will sometimes write $(q,c)$ for a token for the state $q$.

\paragraph{Semantics using configurations}

Let $\aut A$ be an NCA. A \emph{configuration} for $\aut A$ is a set of tokens for $\aut A$. We write $\Cfg(\aut A)$ for the set of all configurations for $\aut A$. Define the configuration transition function $\delta: \Cfg(\aut A) \times \Sigma \to \Cfg(\aut A)$ as follows:
\[
  \delta(S,a) =
  \{ (q,\gamma) \mid
     (p,\beta) \to^a (q,\gamma)
     \text{ for some $(p,\beta) \in S$}
  \}.
\]
We extend the transition function to $\delta: \Cfg(\aut A) \times \Sigma\kstar \to \Cfg(\aut A)$ by $\delta(S,\eps) = S$ and $\delta(S,xa) = \delta(\delta(S,x),a)$ for every $x \in \Sigma\kstar$ and $a \in \Sigma$. Let $S_0$ be the set of all initial tokens, which we call the \emph{initial configuration}, and define $[\aut A]: \Sigma\kstar \to \Cfg(\aut A)$ by $[\aut A](x) = \delta(S_0,x)$.
This semantics coincides with $\sem{\aut A}$ in the following sense: for every $x \in \Sigma\kstar$, $x \in \sem{\aut A}$ iff $[\aut A](x)$ contains some final token.

\paragraph{Bounded counters}

Let $\aut A$ be a NCA, and $n \in \N$ be a constant. We say that a token $(q,\beta)$ is \emph{$n$-bounded} if $\beta(x) \leq n$ for every counter $x \in R(q)$. We also say that $\aut A$ (resp., a state $q$) is \emph{$n$-bounded} if every token (resp., token on state $q$) reachable from some initial token is $n$-bounded. Finally, the NCA $\aut A$ is said to \emph{have bounded counters} if there exists some constant $n \in \N$ such that $\aut A$ is $n$-bounded. Notice that NCAs with bounded counters have the same expressiveness as finite-state automata (i.e., DFAs and NFAs), but they are potentially more succinct \cite{StockmeyerM1973}.

As mentioned earlier, the automata that we consider here are obtained from regexes with counting using the Glushkov construction. A consequence of this is that every counter incrementation action of the form $x\pp$ is guarded by some test $x < n$ because it corresponds to a subexpression of the form $r\{m,n\}$.
It follows that an automaton thus constructed has bounded counters. Moreover, for every control state and every counter, we can read an upper bound from the automaton.
For example, in Figure~\ref{fig:NCA_two_counters}, the counter $x$ is bounded above by $k$ (at all states $q_3, q_4, q_5, q_6$) because $(q_6, \sigma_2, \text{``$x<k$''}, q_3, \text{``$x\pp$''})$ is the only transition that increments $x$. Similarly, the counter $y$ is bounded above by $n$ (at all states $q_4, q_5$) because $(q_5, \sigma_3, \text{``$y<n$''}, q_4, \text{``$y\pp$''})$ is the only transition that increments $y$.

\section{Static Analysis}
\label{sec:static_analysis}

In this section, we will see how to perform a static analysis over regexes to check counter-(un)ambiguity.
It is well-known that the presence of counting in regexes can cause a blow-up in the amount of memory that is needed for the streaming membership problem (checking if a string matches the regex in a single left-to-right pass) \cite{MeyerF1971} (more results about regexes with counting are given in \cite{MeyerS1972, StockmeyerM1973}).
There are, however, many cases that do not exhibit this worst-case behavior. In this section, we will describe a static analysis for identifying occurrences of bounded repetition $\{m,n\}$ which can be implemented using memory that is logarithmic in $n$. This enables a significant reduction in the memory that needs to be reserved for the membership problem.
In order to identify the easier cases of bounded repetition, we use the concept of counter-unambiguity, which informally says that the nondeterminism of the automaton is constrained. We then develop two algorithms for deciding counter-unambiguity (one exact and one approximate), and we provide experimental results showing that they are effective in practice.

Let $\aut A = (Q,R,\Delta,I,F)$ be an NCA. For a state $q \in Q$ and a subset $T \subseteq \Tokens(\aut A)$ of tokens for the automaton, define $T|_q = T \cap (\{q\} \times (R(q) \to \N))$. That is, $T|_q$ contains exactly those tokens of $T$ whose first component is the state $q$. The operational intuition is that $[\aut A](x)|_q$ is the set of tokens that we get at state $q$ when we execute the automaton $\aut A$ on input $x$. When it is possible to have more than two tokens on the same state $q$ after consuming an input string, we say that the state exhibits \emph{counter-ambiguity}. We will now define this concept and other related notions more formally.

\begin{definition}[\bfseries Degree of Counter-Ambiguity]
\label{def:degree_ambiguity}
Let $\aut A$ be an NCA with bounded counters and $q$ be a state. The \emph{(counter-ambiguity) degree} (which we will also call \emph{degree of counter-ambiguity}) of $q$ is defined as
\[
  \degree(q) =
  \tsup_{x \in \Sigma\kstar} \bigl(
    \text{size of $[\aut A](x)|_q$}
  \bigr).
\]
We say that $q$ is \emph{counter-unambiguous} when $\degree(q) \leq 1$, and that $q$ is \emph{counter-ambiguous} when $\degree(q) \geq 2$.
\end{definition}

Notice that if the degree of a state $q$ is equal to zero, then the state $q$ is unreachable.

\subsection{Deciding Counter-Ambiguity}
\label{sec:exact_algo}

According to Definition~\ref{def:degree_ambiguity}, the degree of counter-ambiguity of a state $q$ is the maximum number of different tokens that can end up at $q$ during a computation. A state $q$ is counter-ambiguous iff there is a string $a_1 a_2 \ldots a_n \in \Sigma\kstar$ and two different runs on $a_1 a_2 \ldots a_n$
\begin{gather*}
  (q_0,\beta_0) \xlongrightarrow{a_1}
  (q_1,\beta_1) \xlongrightarrow{a_2}
  (q_2,\beta_2) \xlongrightarrow{a_3}
  \cdots \xlongrightarrow{a_n}
  (q_n,\beta_n)
  \\
  (q'_0,\beta'_0) \xlongrightarrow{a_1}
  (q'_1,\beta'_1) \xlongrightarrow{a_2}
  (q'_2,\beta'_2) \xlongrightarrow{a_3}
  \cdots \xlongrightarrow{a_n}
  (q'_n,\beta'_n),
\end{gather*}
such that $q = q_n = q'_n$ and $\beta_n \neq \beta'_n$.

Let $G$ be the labeled transition system of tokens $\Tokens(\aut A)$ and token transitions of the form $t_1 \to^a t_2$, where $t_1, t_2$ are tokens and $a \in \Sigma$. Define $G^2 = G \times G$ to be the \emph{product} transition system with states $\Tokens(\aut A) \times \Tokens(\aut A)$, which contains a transition $\lt t_1, t_2 \rt \to^a \lt t'_1, t'_2 \rt$ iff $t_1 \to^a t'_1$ and $t_2 \to^a t'_2$. A pair $\lt t_1, t_2 \rt$ is initial if both $t_1$ and $t_2$ are initial tokens. According to the characterization of the previous paragraph, a state $q$ of $\aut A$ is counter-ambiguous iff there exists a path in $G^2$ that ends with some pair $\lt (q,\beta),(q,\beta') \rt$, where $\beta \neq \beta'$. This idea can be extended to characterize the situation where a state $q$ has degree at least $d \geq 2$: there exists a path in the $d$-fold Cartesian product $G^d$ that ends with some tuple $\lt (q,\beta_1), \ldots, (q,\beta_d) \rt$, where $\beta_1, \ldots, \beta_d$ are all distinct.

\textbf{\em Algorithm for Counter-Ambiguity}:
When the product transition system $G^d$ is finite, we can decide whether the counter-ambiguity degree of a state is $\geq d$ with a straightforward reachability algorithm.
For deciding counter-ambiguity, we check whether the degree is $\geq 2$, and therefore it suffices to consider only $G^2$.
Notice that for the bounded counter automata that we consider, $G^d$ is always finite. We just need to exercise care to avoid a blowup in the number of transitions. In our automata, the transitions are annotated with predicates over the alphabet, not symbols of the alphabet. This is a succinct way to represent transitions, and we want to maintain such a representation in the graphs $G^d$ (assuming that we also use such a representation for $G$). This can be done by considering the intersections of predicates and checking whether they are empty. More specifically, for every pair of transitions $t_1 \to^{\sigma_1} t'_1$ and $t_2 \to^{\sigma_2} t'_2$, we add the transition $\lt t_1, t_2 \rt \to^{\sigma_1 \cap \sigma_2} \lt t'_1,t'_2 \rt$ in $G^2$ when $\sigma_1 \cap \sigma_2$ is nonempty.

\begin{example}
\definecolor{lgray}{gray}{0.9}
We will discuss here how to check counter-(un)ambiguity for the regex $\Sigma\kstar \sigma \{2\}$.
First, we construct the NCA for this regex, which is seen below:
\[
\begin{tikzpicture}[->, >=to, auto, node distance=1.5cm, semithick, scale=0.9, transform shape]
\footnotesize
\node (Pre) {};
\node[draw, rounded rectangle, right of=Pre, node distance=1cm] (A) {$q_1$};
\node[draw, rounded rectangle, right of=A, node distance=2cm, accepting] (B) {$q_2: x$};
\node[right of=B] (Post) {};
\path (Pre) edge (A);
\path (A) edge[loop below] node[right, xshift=0.5ex, yshift=0.75ex] {$\Sigma$} (A);
\path (A) edge node {$\sigma \Slash x \coloneqq 1$} (B);
\path (B) edge[loop below] node[right, xshift=0.5ex, yshift=0.75ex] {$\sigma, x<2 \Slash x\pp$} (B);
\path (B) edge node {$x = 2$} (Post);
\end{tikzpicture}
\]
Based on this NCA, we construct the transition system of tokens seen below, where $q_1$ is abbreviation for the token $(q_1, 0_\N)$ ($q_1$ is a pure state), and $(q_2, n)$ is abbreviation for the token $(q_2, x \mapsto n)$ (the counter assignment maps $x$ to $n$).
\[
\begin{tikzpicture}[->, >=to, auto, node distance=1.5cm, semithick, scale=0.9, transform shape]
\footnotesize
\node (Pre) {};
\node[draw, rounded rectangle, right of=Pre, node distance=1cm] (A) {$q_1$};
\node[draw, rounded rectangle, right of=A, node distance=2cm] (B) {$(q_2, 1)$};
\node[draw, rounded rectangle, right of=B, node distance=2cm, accepting] (C) {$(q_2, 2)$};
\path (Pre) edge (A);
\path (A) edge[loop below] node[right, xshift=0.5ex, yshift=0.75ex] {$\Sigma$} (A);
\path (A) edge node {$\sigma$} (B);
\path (B) edge node {$\sigma$} (C);
\end{tikzpicture}
\]
The token transition system is essentially an NFA, where the final state (token) is indicated with a double circle.

To check the counter-ambiguity of a state $q$, we build the product transition system and check whether there exists a path that ends in a pair of tokens $\lt (q,\beta),(q,\beta') \rt$ with $\beta \neq \beta'$.
The figure below shows the product transition system where the presence of the pair $\lt (q_2, 1), (q_2, 2) \rt$ or $\lt (q_2, 2), (q_2, 1) \rt$ (colored in gray) witnesses the counter-ambiguity.
\[
\begin{tikzpicture}[->, >=to, auto, node distance=1.5cm, semithick, scale=0.9, transform shape]
\footnotesize
\node (Pre) {};
\node[draw, rounded rectangle, right of=Pre, node distance=1.5cm] (A) {$\lt q_1,q_1 \rt$};
\node[draw, rounded rectangle, right of=A, node distance=2.5cm] (B) {$\lt q_1, (q_2, 1) \rt$};
\node[draw, rounded rectangle, right of=B, node distance=3.5cm] (C) {$\lt q_1, (q_2, 2) \rt$};
\node[draw, rounded rectangle, below of=A, node distance=1cm] (D) {$\lt (q_2, 1), q_1 \rt$};
\node[draw, rounded rectangle, below of=B, node distance=1cm] (E) {$\lt (q_2, 1), (q_2, 1) \rt$};
\node[draw, rounded rectangle, below of=C, node distance=1cm, fill=lgray] (F) {$\lt (q_2, 1), (q_2, 2) \rt$};
\node[draw, rounded rectangle, below of=D, node distance=1cm] (G) {$\lt (q_2, 2), q_1 \rt$};
\node[draw, rounded rectangle, below of=E, node distance=1cm, fill=lgray] (H) {$\lt (q_2, 2), (q_2, 1) \rt$};
\node[draw, rounded rectangle, below of=F, node distance=1cm, accepting] (I) {$\lt (q_2, 2), (q_2, 2) \rt$};
\path (Pre) edge (A);
\path (A) edge[out=200,in=240,looseness=5] node[left, xshift=-0.5ex, yshift=0.75ex] {$\Sigma$} (A);
\path (A) edge node {$\sigma$} (B);
\path (A) edge node {$\sigma$} (E);
\path (A) edge node {$\sigma$} (D);
\path (B) edge node {$\sigma$} (C);
\path (B) edge node {$\sigma$} (F);
\path (D) edge node {$\sigma$} (G);
\path (D) edge node {$\sigma$} (H);
\path (E) edge node {$\sigma$} (I);
\end{tikzpicture}
\]
Because of symmetry, some states and transitions can be safely removed from the product automaton. Notice, for example, that we do not need to explore both $\lt (q_2, 1), q_1 \rt$ and $\lt q_1, (q_2, 1) \rt$. Therefore, in future examples, we will omit part of the product automaton.
\end{example}

The exact analysis halts as soon as it finds a token pair that witnesses counter-ambiguity. So, not all pairs are generated during the static analysis, unless the regex is counter-unambiguous. 

Consider a regex $r$ that contains an occurrence of counting of the form $(abcd)\{m,n\}$. When the repetition bounds are sufficiently large, in the automaton $\aut A$ for $r$, the four states that correspond to $abcd$ are either all counter-unambiguous or they are all counter-ambiguous. For this reason, the notion of counter-(un)ambiguity can be defined with respect to instances of bounded repetition in regexes.
%
%
We will also call a regex counter-ambiguous if it contains at least one occurrence of bounded repetition that is counter-ambiguous (equivalently, the NCA for the expression has at least one counter-ambiguous state).

\begin{lemma}[\bfseries Checking Counter-Ambiguity Is Hard]
\label{lemma:hardness}
\normalfont
Let \textsc{CAmbiguity} be the following problem: Given a regex $r$ as input, is $r$ counter-ambiguous? \textsc{CAmbiguity} is NP-hard.
\end{lemma}
\begin{proof}
Consider the alphabet $\Sigma = \{ a, b, \# \}$. We will give a polynomial-time reduction from the subset sum problem to \textsc{CAmbiguity}. Let $S = \{ n_1, n_2, \ldots n_m \}$ be a set of natural numbers and $T$ be a natural number. Recall that the subset sum problem asks whether there is a subset $S' \subseteq S$ of numbers whose sum is equal to $T$. Consider the regex
\[
  ( ((a\{n_1\} + \eps) \cdots (a\{n_m\} + \eps) \# b) + (a\{T\} \# bb) )b\{2\}.
\]
We focus on the rightmost occurrence of bounded repetition (i.e., $b\{2\}$). We claim that this occurrence is counter-ambiguous if and only if there is a subset $S' \subseteq S$ whose sum is $T$. Consider the corresponding Glushkov automaton and the state $q$ which leads to the final state at the end that recognizes the $b\{2\}$. A word witnessing a path to $q$ would have to be of the form $a^x \# b^y$ for some natural numbers $x, y$. If $x \neq T$, then the word has no path through the branch $(a\{T\} \# bb)$. So, the only value it can induce on the counter at the end is $(y - 2)$. If $x= T$, and there exists a subset $S'$ of $S$ such that $\sum S' = T$, then $a\{T\} \# bbb$ could either take the path $(a\{T\}\#bb)$ and set the counter to 1, or it could take the other path and set the counter to 2. If $x = T$ and there is no such subset $S'$, then the only path the word can take is through the branch $(a\{T\} \# bb)$ which would set the counter to $(y - 2)$.
\end{proof}

\subsection{Over-Approximate Analysis}
\label{sec:approx_algo}

In \S\ref{sec:exact_algo}, we presented an (exact) algorithm for deciding the counter-(un)ambiguity of regexes and NCAs. The algorithm operates on the transition system of tokens of an NCA, whose size can be exponential in the size of the regex, because of the counter valuations.
For example, the regex $\Sigma\kstar \cdot a \cdot \Sigma\{n\}$ has size $\Theta(\log n)$ (because the repetition bound $n$ is represented succinctly in binary or decimal notation) and the corresponding token transition system has size $\Theta(n)$.
From this it follows that the exact algorithm may need exponential time in the worst case. Unfortunately, this worst-case behavior is not easy to avoid given the NP-hardness of the problem (Lemma~\ref{lemma:hardness}).
For this reason, we propose here a heuristic algorithm that performs an ``over-approximate'' analysis, which can give two outputs: it either declares that a state is counter-unambiguous, or it says that the analysis is inconclusive. In other words, there are cases where the algorithm may suspect that a state is counter-ambiguous, but it cannot conclusively declare it so.

The idea is to over-approximate all occurrences of $\{m,n\}$ (constrained repetition) with $\kstar$ (unconstrained repetition), except for the one that we are analyzing. If we think of this transformation in terms of NCAs, we see that it adds more paths to the token transition graph, because more transitions are now enabled. A consequence of this is that if the over-approximate automaton is counter-unambiguous, then surely the original automaton (which has less paths) is also counter-unambiguous. On the other hand, if the over-approximate automaton is counter-ambiguous, then we cannot infer that the original automaton is counter-ambiguous.

\begin{example}
\label{exp:approx}
\definecolor{lgray}{gray}{0.9}
We show the static analysis for a counter-unambiguous regex $r = \Sigma\kstar(\bar\sigma_1 \sigma_1 \{n\} + \bar\sigma_2 \sigma_2 \{n\})$, where $n$ is a constant. For this regex, the over-approximate analysis is more efficient than the exact analysis. To illustrate this, we first construct the NCA:
\[
\begin{tikzpicture}[->, >=to, auto, node distance=1.5cm, semithick, scale=0.9, transform shape]
\footnotesize
\node (Pre) {};
\node[draw, rounded rectangle, right of=Pre, node distance=1cm] (A) {$q_1$};
\node[draw, rounded rectangle, right of=A, node distance=2cm] (B) {$q_2$};
\node[draw, rounded rectangle, below of=B, node distance=1cm] (C) {$q_3$};
\node[draw, rounded rectangle, right of=B, node distance=2.5cm, accepting] (D) {$q_4: x$};
\node[draw, rounded rectangle, right of=C, node distance=2.5cm, accepting] (E) {$q_5: x$};
\node[right of=D] (Post1) {};
\node[right of=E] (Post2) {};
\path (Pre) edge (A);
\path (A) edge[loop below] node {$\Sigma$} (A);
\path (A) edge node {$\bar \sigma_1$} (B);
\path (A) edge[bend right] node {$\bar \sigma_2$} (C);
\path (B) edge node {$\sigma_1 \Slash x \coloneqq 1$} (D);
\path (D) edge[loop below] node[right, xshift=0.5ex, yshift=1.5ex] {$\sigma_1, x < n \Slash x\pp$} (D);
\path (C) edge node {$\sigma_2 \Slash x \coloneqq 1$} (E);
\path (E) edge[loop below] node[right, xshift=0.5ex, yshift=1.5ex] {$\sigma_2, x< n \Slash x\pp$} (E);
\path (D) edge node {$x = n$} (Post1);
\path (E) edge node {$x = n$} (Post2);
\end{tikzpicture}
\]
The exact analysis constructs the token transition system:
\[
\begin{tikzpicture}[->, >=to, auto, node distance=1.5cm, semithick, scale=0.9, transform shape]
\footnotesize
\node (Pre) {};
\node[draw, rounded rectangle, right of=Pre, node distance=1cm] (A) {$q_1$};
\node[draw, rounded rectangle, right of=A, node distance=1.5cm] (B) {$q_2$};
\node[draw, rounded rectangle, below of=B, node distance=0.8cm] (C) {$q_3$};
\node[draw, rounded rectangle, right of=B, node distance=1.7cm] (D) {$(q_4, 1)$};
\node[draw, rounded rectangle, right of=C, node distance=1.7cm] (E) {$(q_5, 1)$};
\node[right of=D, node distance=1.5cm] (F) {...};
\node[right of=E, node distance=1.5cm] (G) {...};
\node[draw, rounded rectangle, right of=F, node distance=1.7cm, accepting] (H) {$(q_4, n)$};
\node[draw, rounded rectangle, right of=G, node distance=1.7cm, accepting] (I) {$(q_5, n)$};
\path (Pre) edge (A);
\path (A) edge[loop below] node[left, xshift=-0.5ex, yshift=0.5ex] {$\Sigma$} (A);
\path (A) edge node {$\bar \sigma_1$} (B);
\path (A) edge[bend right] node[swap, xshift=0.5ex, yshift=0.5ex] {$\bar \sigma_2$} (C);
\path (B) edge node {$\sigma_1$} (D);
\path (C) edge node {$\sigma_2$} (E);
\path (D) edge node {$\sigma_1$} (F);
\path (E) edge node {$\sigma_2$} (G);
\path (F) edge node {$\sigma_1$} (H);
\path (G) edge node {$\sigma_2$} (I);
\end{tikzpicture}
\]
To determine whether the regex is counter-unambiguous, the exact analysis explores all possible token pairs in the product transition system. In this example, the number of explored pairs is $\Theta(n^2)$. Below is a part of the product transition system, in which all token pairs $\lt (q_5, i), (q_4, j)\rt$ with $1 \leq i < j \leq n$ (colored in gray) will be explored.
\[
\begin{tikzpicture}[->, >=to, auto, node distance=1.8cm, semithick, scale=0.9, transform shape]
\footnotesize
\node (Pre) {};
\node[draw, rounded rectangle, right of=Pre, node distance=1.5cm] (A) {$\lt q_1, q_1 \rt$};
\node[draw, rounded rectangle, right of=A, node distance=2.5cm] (B) {$\lt q_1, q_2 \rt$};
\node[draw, rounded rectangle, below of=A, node distance=1.25cm] (C) {$\lt q_3, (q_4, 1) \rt$};
\node[draw, rounded rectangle, right of=B, node distance=2.5cm] (D) {$\lt q_1, (q_4, 1) \rt$};
\node[draw, rounded rectangle, below of=C, node distance=1.25cm, fill=lgray] (E) {$\lt (q_5, 1), (q_4, 2) \rt$};
\node[draw, rounded rectangle, below of=B, node distance=1.25cm] (F) {$\lt q_3, (q_4, 2) \rt$};
\node[draw, rounded rectangle, right of=F, node distance=3cm, fill=lgray] (G) {$\lt (q_5, 1), (q_4, 3) \rt$};
\node[draw, rounded rectangle, right of=E, node distance=3.5cm, fill=lgray] (H) {$\lt (q_5, 2), (q_4, 3) \rt$};
\node[right of=D] (Post1) {...};
\node[right of=H] (Post2) {...};
\node[right of=G] (Post3) {...};
\path (Pre) edge (A);
\path (A) edge[loop below] node[left, xshift=-0.5ex, yshift=0.5ex] {$\Sigma$} (A);
\path (A) edge node {$\bar \sigma_1$} (B);
\path (B) edge node[pos=0.2] {$\bar \sigma_2 \cap \sigma_1$} (C);
\path (B) edge node {$\sigma_1$} (D);
\path (C) edge node {$\sigma_2 \cap \sigma_1$} (E);
\path (D) edge node[pos=0.2] {$\bar \sigma_2 \cap \sigma_1$} (F);
\path (F) edge node[swap] {$\sigma_2 \cap \sigma_1$} (G);
\path (E) edge node {$\sigma_2 \cap \sigma_1$} (H);
\path (D) edge node {} (Post1);
\path (H) edge node {} (Post2);
\path (G) edge node {} (Post3);
\end{tikzpicture}
\]
We have observed that regexes of the form $r = \Sigma\kstar(\bar \sigma_1 \sigma_1 \{n\} + \bar\sigma_2 \sigma_2 \{n\})$, where $n$ is a large number, can be found in the Snort and Suricata benchmarks. For these regexes, the exact analysis may require a long computation.
Fortunately, the over-approximate analysis is substantially faster. We approximate the regex as $r' = \Sigma\kstar(\bar \sigma_1 \sigma_1 \{n\} + \bar\sigma_2 \sigma_2 \kstar)$ and $r'' = \Sigma\kstar(\bar \sigma_1 \sigma_1\kstar + \bar\sigma_2 \sigma_2 \{n\})$ and check the counter-ambiguity of $r'$ and $r''$ using the exact analysis.
The regex $r$ is determined to be counter-unambiguous if both $r'$ and $r''$ are counter-unambiguous.
Below, we construct the token transition system $G$ for $r'$. Only $\Theta(n)$ token pairs are explored in the product transition system $G^2$.
\[
\begin{tikzpicture}[->, >=to, auto, node distance=1.5cm, semithick, scale=0.9, transform shape]
\footnotesize
\node (Pre) {};
\node[draw, rounded rectangle, right of=Pre, node distance=1cm] (A) {$q_1$};
\node[draw, rounded rectangle, right of=A, node distance=1.3cm] (B) {$q_2$};
\node[draw, rounded rectangle, below of=B, node distance=0.6cm, accepting] (C) {$q_3$};
\node[draw, rounded rectangle, right of=B, node distance=1.7cm] (D) {$(q_4, 1)$};
\node[right of=D, node distance=1.7cm] (F) {...};
\node[draw, rounded rectangle, right of=F, node distance=1.7cm, accepting] (H) {$(q_4, n)$};
%
\path (Pre) edge (A);
\path (A) edge[loop below] node[left, xshift=-0.5ex] {$\Sigma$} (A);
\path (A) edge node {$\bar \sigma_1$} (B);
\path (A) edge[bend right] node[swap, xshift=1ex, yshift=0.5ex] {$\bar \sigma_2$} (C);
\path (C) edge[loop right] node {$\sigma_2$} (C);
\path (B) edge node {$\sigma_1$} (D);may introduce some additional paths in the token transition system
\path (D) edge node {$\sigma_1$} (F);
\path (F) edge node {$\sigma_1$} (H);
%
\end{tikzpicture}
\]
The over-approximate analysis checks the counter-ambiguity of $r', r''$. So, it reduces the complexity from $\Theta(n^2)$ to $\Theta(n)$.
\end{example}

\subsubsection{NCA Execution with Bit Vectors}
\label{sec:nca_with_bit_vectors}

If the static analysis determines that an NCA state $q$ is counter-ambiguous, then this implies that the execution of the automaton may require several memory locations to store tokens of the form $(q,\beta)$. Assuming that $q$ has only one counter register $x$ (i.e., $R(q) = \{x\}$) and that $q$ is $n$-bounded, we know that there are at most $n$ different possible tokens. In order to compactly represent a set of tokens, the idea is to use a bit vector that indicates the presence or the absence of a specific token on $q$. So, a bit vector $v$ encodes a set of tokens on $q$ as follows: $v[i] = 1$ iff the token $(q,i)$ is active. We can also think of a bit vector as a representation for part of the automaton configuration (recall the configuration semantics from \S\ref{sec:prelim}).

It remains to see how the execution of the automaton can be described using these bit vectors to represent the configuration. Example~\ref{ex:NCAs} shows the NCA for the regex $\Sigma\kstar \sigma_1 (\sigma_2 \sigma_3)\{m,n\} \sigma_4$. This NCA is general enough to illustrate the main ways in which we manipulate bit vectors:
\begin{compactenum}[(1)]
\item
Consider a transition $p \to q$, annotated with ``$\sigma \Slash x \coloneqq c$'', where $p$ is pure and $R(q) = \{ x \}$. A token on $p$ is transformed into a bit vector $v$ for $q$ that is everywhere 0 except that $v[c] = 1$.
\item
Let $p \to q$ be a transition, annotated with $\sigma$, where $R(p) = R(q) = \{x\}$. Since the transition does not change the counter valuations, a bit vector $v$ on $p$ is passed along unchanged to $q$.
\item
We will deal now with a transition $p \to q$, annotated with ``$\sigma, x < n \Slash x\pp$'', where $R(p) = R(q) = \{x\}$. Assume further that both $p$ and $q$ are $n$-bounded, which means that each state carries a bit vector of size $n$. This transition corresponds to performing a \emph{shift operation} to the bit vector $v$ of $p$, resulting in a new bit vector $v'$ for $q$. We have: $v'[1] = 0$ and $v'[i + 1] = v[i]$ for ever $i = 2, \ldots, n-1$.
\item
Finally, let us consider a transition $p \to q$, annotated with ``$\sigma, m \leq x \leq n$'', where $R(p) = \{ x \}$ and $q$ is pure. If $v$ is the current bit vector for $p$, then taking this transition produces a token for $q$ if and only if one of $v[m], v[m+1], \ldots, v[n-1], v[n]$ is equal to 1. In other words, we have to compute the disjunction $v[m] \lor \cdots \lor v[n]$.
\end{compactenum}
The above cases involve the main operations that we use for bit vectors: setting the least significant bit (case 1), shifting left by one position (case 3), and computing the disjunction of some of the most significant bits (case 4).

The way bit vectors are used (setting the lowest-order bit, shifting, and reading high-order bits) is similar to how queues and sliding windows are used for runtime verification with metric temporal logic (MTL) \cite{BartocciDDFMNS2018, MamourasW2020, ChattopadhyayM2020, MamourasCW2021TACAS}. We note that MTL involves constructs that specify time durations with intervals of the form $[m,n]$, which are akin to the bounded repetition operators $\{m,n\}$ of regexes. This explains the similarity in the implementation.

\subsection{Implementation and Experiments}
\label{sec:static_analysis_performance}

We have implemented a Java program that statically analyzes regexes to determine if they are counter-(un)ambiguous. We will call this program the \emph{counter-ambiguity checker}. The implementation includes both the exact and the over-approximate analyses.
As the approximate analysis may be unable to verify the counter-ambiguity of some instances, our checker implements a \textbf{\em hybrid analysis}. First, it checks the counter-(un)ambiguity of each instance of bounded repetition in the regex using the over-approximate analysis. If it finds a potentially counter-ambiguous instance, then it halts the over-approximate analysis and uses the exact algorithm to check the regex. Otherwise, it determines that the regex is counter-unambiguous.

The checker not only determines if a regex is counter-ambiguous but also provides a \emph{counter-ambiguity witness}, which is a string over the alphabet. If the NCA is executed on the witness, then at least two tokens with different counter valuations will end up on some state of the NCA.
The checker supports the analysis of counter-ambiguity for each instance of bounded repetition inside a regex. For example, given a regex $\sigma_1\{m\}\Sigma\kstar \sigma_2\{n\}$, it can check the first instance (i.e., $\{m\}$), which is counter-unambiguous, and the second instance (i.e., $\{n\}$), which is counter-ambiguous.

\begin{table}
\caption{Analysis of regexes in the benchmarks.}
\label{table:benchdata}
\footnotesize
\centering
\renewcommand{\arraystretch}{1.1}
\begin{tabular}{|crrrr|} 
\hline
Benchmark & \multicolumn{1}{c}{\# total} & \multicolumn{1}{c}{\# supported} & \multicolumn{1}{c}{\# counting} & \multicolumn{1}{c|}{\# c-ambiguous}
\\ \hline
Protomata & 2338 & 2338 & 1675 & 1675
\\
Snort & 5839 & 5315 & 1934 & 282
\\
Suricata & 4480 & 3728 & 1510 & 246
\\
SpamAssassin & 3786 & 3690 & 459 & 279
\\
ClamAV & 100472 & 100472 & 4823 & 3626
\\ \hline
\end{tabular}
\end{table}

We evaluate the performance of our counter-ambiguity checker using five benchmarks, which contain regexes collected from real applications. These benchmarks are: (1) the \textbf{Snort} \cite{Snort} and (2) \textbf{Suricata} benchmarks \cite{Suricata} that contain patterns for network traffic, (3) the \textbf{Protomata} benchmark that includes 1309 protein motifs from the PROSITE database \cite{Prosite, RoyA2016}, (4) the \textbf{ClamAV} benchmark \cite{ClamAV} that contains patterns that indicate the presence of viruses, and (5) the \textbf{SpamAssassin} benchmark \cite{SpamAssassin} that includes patterns for detecting spam email.

Table~\ref{table:benchdata} shows some statistics for the regexes included in the benchmarks.
In the Snort, Suricata, and SpamAssassin benchmarks, some of the collected regexes may contain backreferences \cite{PosixSyntax}, which is not a regular operator (i.e., it can give rise to non-regular languages).
We filter out regexes with backreferences from the datasets and perform the static analysis on the remaining regexes (which contain the supported regular operators).
Table~\ref{table:benchdata} provides the following information: the total number of regexes for each benchmark, the number of regexes with supported (regular) operators, the number of regexes with at least one occurrence of constrained repetition (counting), and the number of counter-ambiguous regexes.

\begin{figure*}
\centering
\begin{subfigure}{0.5\linewidth}
\centering
\includegraphics[width=\linewidth]{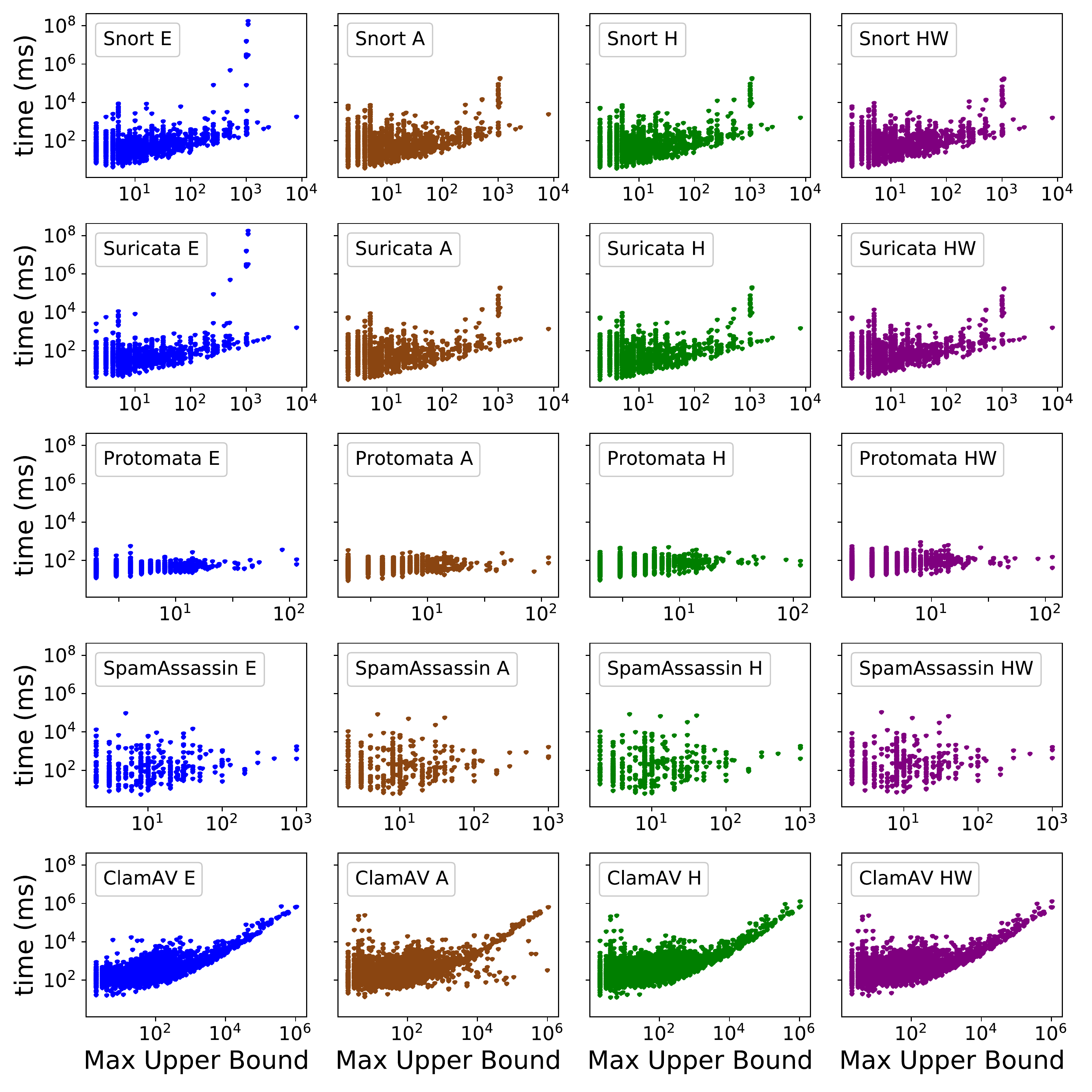}
\caption{running time}
\end{subfigure}%
\begin{subfigure}{0.5\linewidth}
\centering
\includegraphics[width=\linewidth]{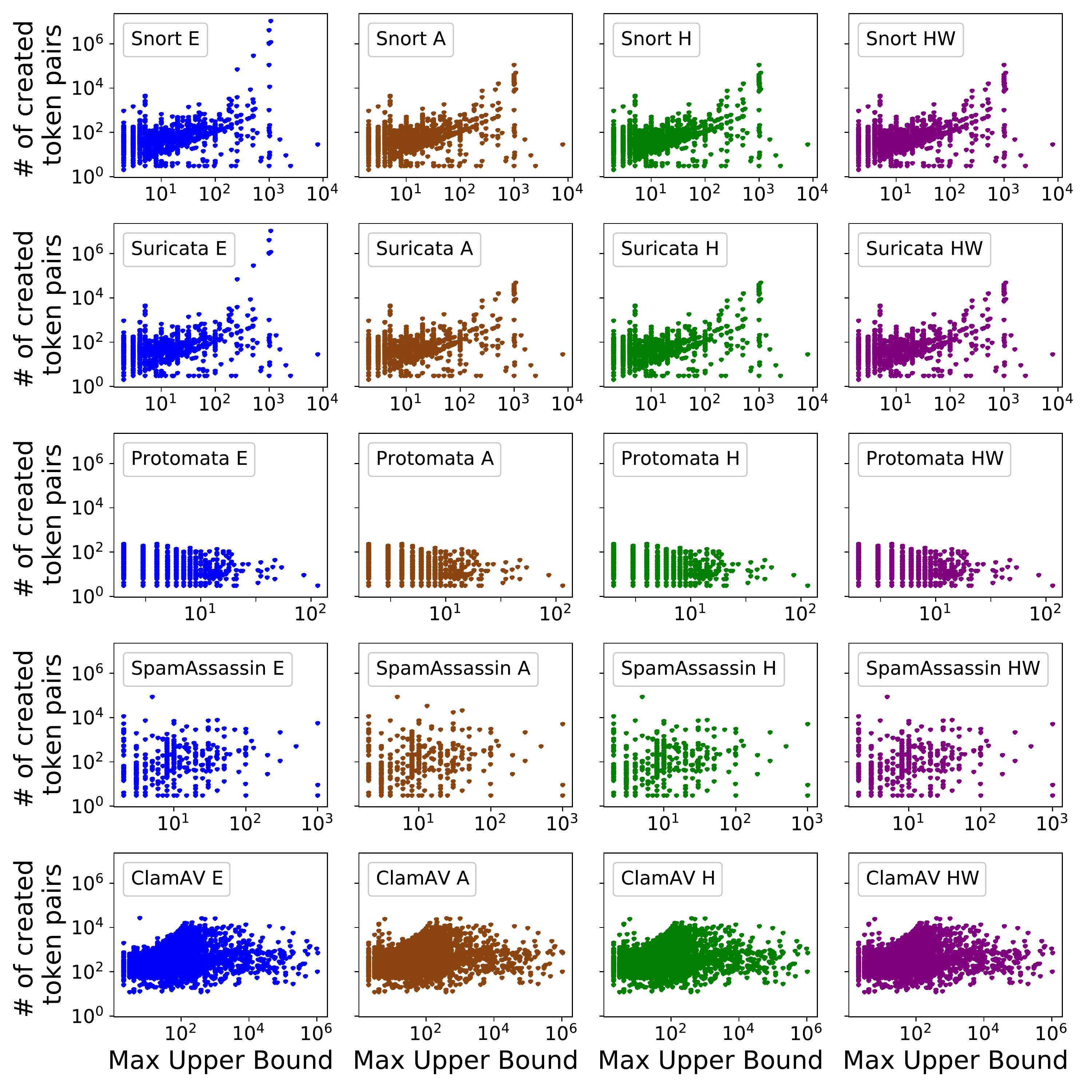}
\caption{\# of created token pairs}
\end{subfigure}
\caption{The (a) running time and the (b) \# of created token pairs of static analysis for regexes with different maximum upper bounds of repetitions. E means exact analysis, A means approximate analysis, H means hybrid analysis, HW means hybrid analysis with reporting inputs that witness the ambiguity. E.g., ``Snort E'' means the exact analysis in Snort benchmark.}
\label{fig:static-duration}
\label{fig:static-space}
\end{figure*}

\begin{figure}
\centering
\includegraphics[width=0.92\linewidth]{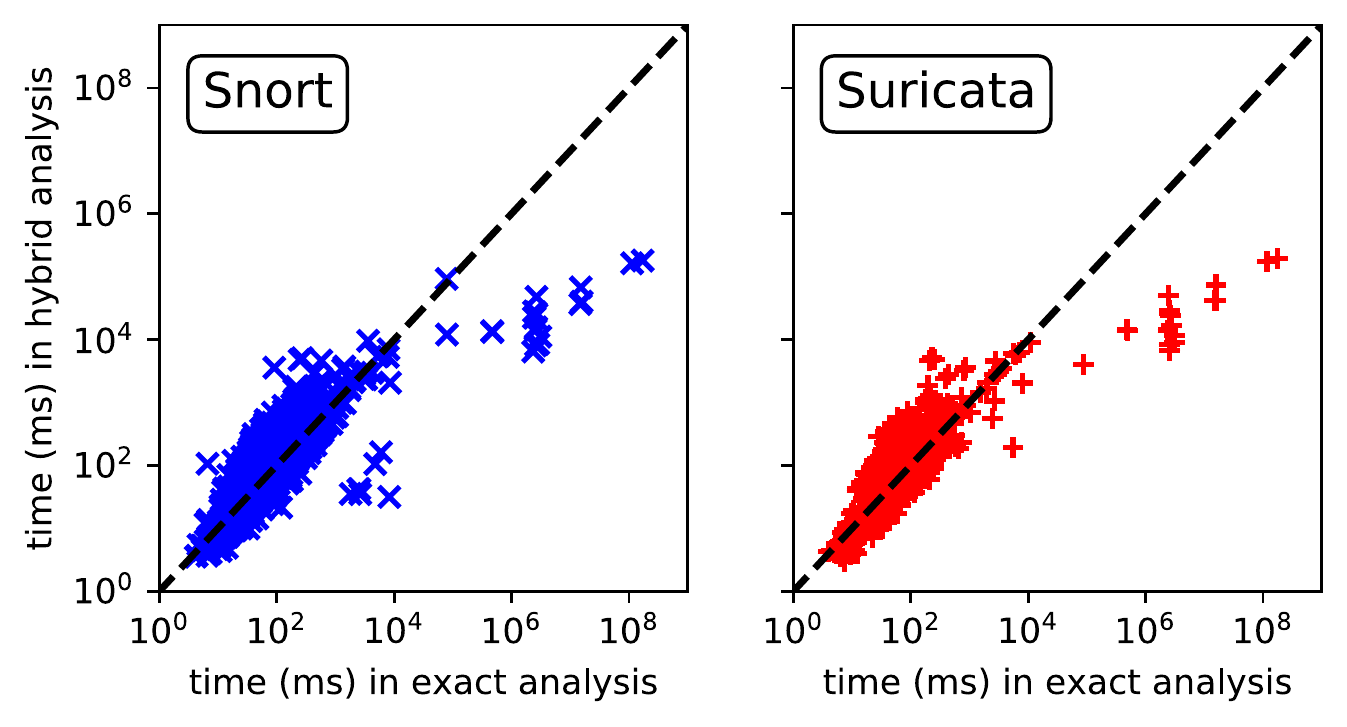}
\caption{Running time (ms) comparison of exact and hybrid analyses on the Snort and Suricata benchmarks.}
\label{fig:outliers-zoom-in}
\end{figure}

\paragraph{Experimental setup}

The experiments were executed in Ubuntu 20.04 on a desktop computer equipped with an Intel Xeon(R) E3-1241 v3 CPU (4 cores) with 16 GB of memory (DDR3 at 1600 MHz). We used OpenJDK 17 and set the maximum heap size to 4 GB. For each regex, we executed 20 trials and selected the mean runtime as the value used the reported results (excluding the first 10 ``warm-up'' trials).

\paragraph{Performance: Running Time.}

We evaluate the performance of the static analysis over regexes that have non-nested instances of constrained repetition.
We report the running time of the static analysis and we consider its dependence on the following ``measure of complexity'' for a regex $r$: the maximum repetition upper bound over all occurrences of $\{m,n\}$ in a regex, which we denote by $\mu(r)$. For example, the regex $r = \sigma_1\{1,5\}\sigma_2\sigma_3\{4\}$ has two occurrences of constrained repetition, and the maximum repetition upper bound is $\mu(r) = \max(5,4) = 5$.
In general, we expect the running time for the analysis of a regex $r$ to depend on $\mu(r)$, since checking counter-ambiguity involves the generation of token pairs whose number increases as $\mu(r)$ increases.

Figure~\ref{fig:static-duration}(a) shows the running time of the static analysis indexed by the measure $\mu$. The results are shown in 20 plots, which are organized in a $5 \times 4$ grid. There are 5 rows, one for each benchmark: Snort, Suricata, Protomata, SpamAssassin, ClamAV. There are 4 columns, one for each variant of the static analyzer: exact, approximate, hybrid, and hybrid with witness reporting.
Each of these 20 plots contains multiple points, one for each regex of the benchmark. For every regex $r$, the corresponding point has horizontal coordinate equal to $\mu(r)$ and vertical coordinate equal to the running time of the analysis (in milliseconds).
We observe that the running time for analyzing a regex $r$ generally increases as $\mu(r)$ increases.

In the Snort and Suricata benchmarks, the checker takes more than 100 seconds to perform the exact analysis for several counter-unambiguous regexes. See the top-right outliers in the plots labeled ``Snort E'' and ``Suricata E'' in Figure~\ref{fig:static-duration}(a).
This information is seen more promimently in Figure~\ref{fig:outliers-zoom-in}, where the exact and hybrid analyses are compared on the Snort and Suricata benchmarks. The points with horizontal coordinate ${>} 10^5$ (msec) are noteworthy. They are substantially below the diagonal, which means that the hybrid analysis offers significant improvement in terms of running time.
Some of these regexes are of the form $\Sigma\kstar (\bar \sigma_1 \sigma_1\{m\} + \bar \sigma_2 \sigma_2\{n\} + \cdots)$, where $m, n, \ldots$ are large numbers.
When performing exact analysis on these regexes, the checker needs to explore a large number of token pairs, which makes the analysis time-consuming.
However, as discussed in Example~\ref{exp:approx}, the over-approximate analysis can greatly reduce the cost of the computation.
We observe that the over-approximate analysis reduces the running time of expensive regexes by over 100 times in both the Snort and Suricata benchmarks.
Moreover, as these regexes are counter-unambiguous, the result of their over-approximate analysis is accurate. This explains why the hybrid analysis also reduces the running time of these challenging regexes.

The fourth column in Figure~\ref{fig:static-duration}(a) shows the performance (in terms of running time) of a variant of the static analyzer that reports a witness (input string) when a regex is counter-ambiguous.
We observe that finding and reporting a counter-ambiguity witness add a very small overhead to the static analysis. This is because recording the witness amounts to simply storing a transition symbol whenever the analysis moves from one token pair to another.

\paragraph{Performance: Memory Footprint.}

The checker analyzes the counter-ambiguity of a regex by exploring token pairs in a product transition system. These token pairs are created on the fly, as the transition system is being explored.
We estimate the memory footprint of the static analysis by measuring the number of token pairs that the checker creates.
Figure~\ref{fig:static-space}(b) shows the results for five benchmarks and four different variants of the static analysis. Similarly to the case of running time, the over-approximate analysis greatly reduces the worst-case cost of analyzing several counter-unambiguous regexes in the Snort and Suricata benchmarks.

\section{Hardware Implementation and Experiments}
\label{sec:implementation}

In this section, we present our hardware design for efficiently executing NCAs. We augment a state-of-the-art in-memory NFA acceleration architecture called CAMA \cite{cama} with counter and bit vector modules. We report hardware simulation results in both microbenchmarks and application benchmarks.

\subsection{Hardware Design}

Existing in-memory automata accelerators adopt a two-phase architecture: a state matching phase that finds the current active states, and a state transition phase that calculates the available states in the next cycle.
AP-style accelerators, such as AP \cite{DlugoschBGLN2014AP}, CA \cite{SubramaniyanAW17}, and eAP\footnote{eAP stands for embedded Automata Processor.} \cite{SRVSS2019eAP}, perform state matching by reading from read-access memories (RAMs) that store bit vector representations of states in memory columns. Each column in the RAM represents one state, which is called a State Transition Element (STE). 
Using 8-bit symbols as an example, each RAM entry is 256-bit and the $i$-th position has value 1 iff the symbol $i$ is associated with the state\footnote{Recall from \S\ref{sec:prelim} that we consider homogeneous automata, which means that all transitions leading to a state $q$ are labeled with the same predicate $\sigma$ over the alphabet. The RAM entry is a representation of the predicate $\sigma$.}. 
Additionally, the connections between states are programmed into a switch network where existing state transitions are realized as physical connections.

Each processing cycle begins in the state matching phase, where an input symbol is encoded as a one-hot representation\footnote{The one-hot representation of an 8-bit symbol $i$ consists of $2^8 = 256$ bits, where the $i$-th bit has value 1 and the others are 0.} and used as the address to read from the state matching memory.
The columns that read out `1's indicate successful matches between the input symbol and the STEs. 
With a logical AND operation between the available states reported from the last cycle and the matched states reported by the memory in the current cycle, matching results of the active states in the current cycle are determined.
Next, in the state transition phase, the current active states pass through the programmed switch network to create the next vector which stores available states for the next cycle.


However, AP-style accelerators severely under-utilize the state matching memories in realistic NFAs across common benchmarks, because this approach is optimal only for the worst case of purely random NFAs. Impala \cite{SadrediniER20} and CAMA \cite{cama} made critical improvements by proposing special encoding schemes to reduce the state matching memory requirements. CAMA further employs specialized content-addressable memories (CAM) to perform state matching with lower energy and memory footprints than all other designs using RAM. As a result, the memory requirement for 256 STEs is reduced from one 256$\times$256 6-transistor SRAM in AP and CA, to two 16$\times$256 6-transistor SRAMs in Impala and approximately one 16$\times$256 8-transistor CAM in CAMA.
Moreover, CAMA optimizes a reduced-crossbar switch network that was first proposed by eAP, which largely reduces the area and energy costs of state transitions.
Compared with prior NFA in-memory architectures, CAMA achieves leading throughput, energy, and area efficiency. 
CAMA's throughput is 2.14GBps, 1.18x better than CA, 9.5x better than FPGA-based Grapefruit \cite{RahimiRS20}, and 2-4 orders better than CPU/GPU solutions. 
CAMA's energy efficiency is 4.91nJ/Byte, over 10x better than most efficient alternatives, i.e. Grapefruit (FPGA) and AP.
This paper uses the latest memory- and energy-efficient CAMA architecture as the baseline and augments it with our proposed counter and bit vector modules.

Figure~\ref{fig:hw-circuit}(a) shows the Glushkov NCA for the counter-unambiguous regex $a(bc)\{1,3\}c$. The Glushkov construction ensures that the NCA is homogeneous (all transitions entering a state are labeled with the same predicate over the alphabet). This property allows us to convert the NCA to a hardware-friendly representation by omitting the initial state and pushing the predicates from the edges to the states, thus transforming NCA states into STEs.
For example, we push the predicate $a$ into state $q_a$ so that in Figure~\ref{fig:hw-circuit}(b) we have a state labeled with the predicate $a$, which becomes an STE that is activated to fire signals only when the input satisfies the predicate $a$.
The original CAMA design, as shown in Figure~\ref{fig:hw-circuit}(c), only supports NCAs by fully unfolding bounded repetitions.
In our augmented CAMA, two types of hardware modules, counters and bit vectors, are added to accelerate the execution of NCAs. As shown in Figure~\ref{fig:hw-circuit}(d), both modules take input from STEs related to counting and produce output signals to the switch network. 
Counters are inserted to support counter-unambiguous repetitions, while bit vectors are reserved for counter-ambiguous repetitions (recall \S\ref{sec:nca_with_bit_vectors}). 
Compared to CAMA, the additional counters and bit vectors retain all necessary processing information while avoiding the cost of unfolding (which results in additional STEs).
In Section~\ref{sec:toMNRL}, we will further explain the design and the input/output ports of the counter and bit vector modules.

Figure~\ref{fig:hw-arch} shows the structure of an augmented CAMA bank. The overall architecture of CAMA is preserved, and the functionalities of existing components remain the same. Each bank consists of an input/output buffer and 16 processing arrays. Each array has a global switch and 8 processing elements (PEs). Each PE contains two 256-STE CAM arrays, two local switches, and 8 counters, and it may contain a bit vector depending on the configuration from users.
Note that the input ports to the counter and bit vector modules are connected to fixed groups of STEs. For example, as shown on the right, port \texttt{pre} is connected to STEs 0 to 7, port \texttt{fst} is connected to STEs 8 to 16, and so on. When enabled, an STE within the group can pass signals to the connected port. 
We use an efficient mapping algorithm to build the connection between ports and STE groups so that we maintain the generality of the design but reduce the complexity of routing.

It is worth mentioning that our proposed counters and bit vectors are not only suitable for the CAMA architecture. Other in-memory automata architectures, like CA, can also be augmented for NCAs with minor hardware design changes.
Specifically, these changes are: (1) counters and bit vectors need to be allowed to connect to elements that represent states, and (2) the routing network needs to be extended to store the transitions from counters and bit vectors.

\begin{figure}
\centering
\begin{subfigure}{0.99\linewidth}
    \centering
    \begin{tikzpicture}[->, >=to, auto, node distance=1cm, semithick, scale=0.9, transform shape]
    \footnotesize
    \node (Pre) {};
    \node[draw, rounded rectangle, right of=Pre, node distance=1cm] (A) {$q_0$};
    \node[draw, rounded rectangle, right of=A, node distance=1cm] (B) {$q_a$};
    \node[draw, rounded rectangle, right of=B, node distance=2cm] (C) {$q_b: x$};
    \node[draw, rounded rectangle, right of=C, node distance=2.25cm] (D) {$q_c: x$};
    \node[draw, rounded rectangle, right of=D, node distance=2cm, accepting] (E) {$q_d$};
    \node[right of=E] (Post1) {};
    \path (Pre) edge (A);
    \path (A) edge node {$a $} (B);
    \path (B) edge node {$b \Slash x \coloneqq 1$} (C);
    \path (C) edge[bend left=10] node {$c$} (D);
    \path (D) edge[bend left=10] node {$b, x<3 \Slash x\pp$} (C);
    \path (D) edge node {$d, 1 \leq x \leq 3$} (E);
    \path (E) edge node {} (Post1);
    \end{tikzpicture}
    \\
    \vspace{-1mm}
    {\small (a)}
\end{subfigure}
\begin{subfigure}{\linewidth}
    \centering
    \includegraphics[width=\linewidth]{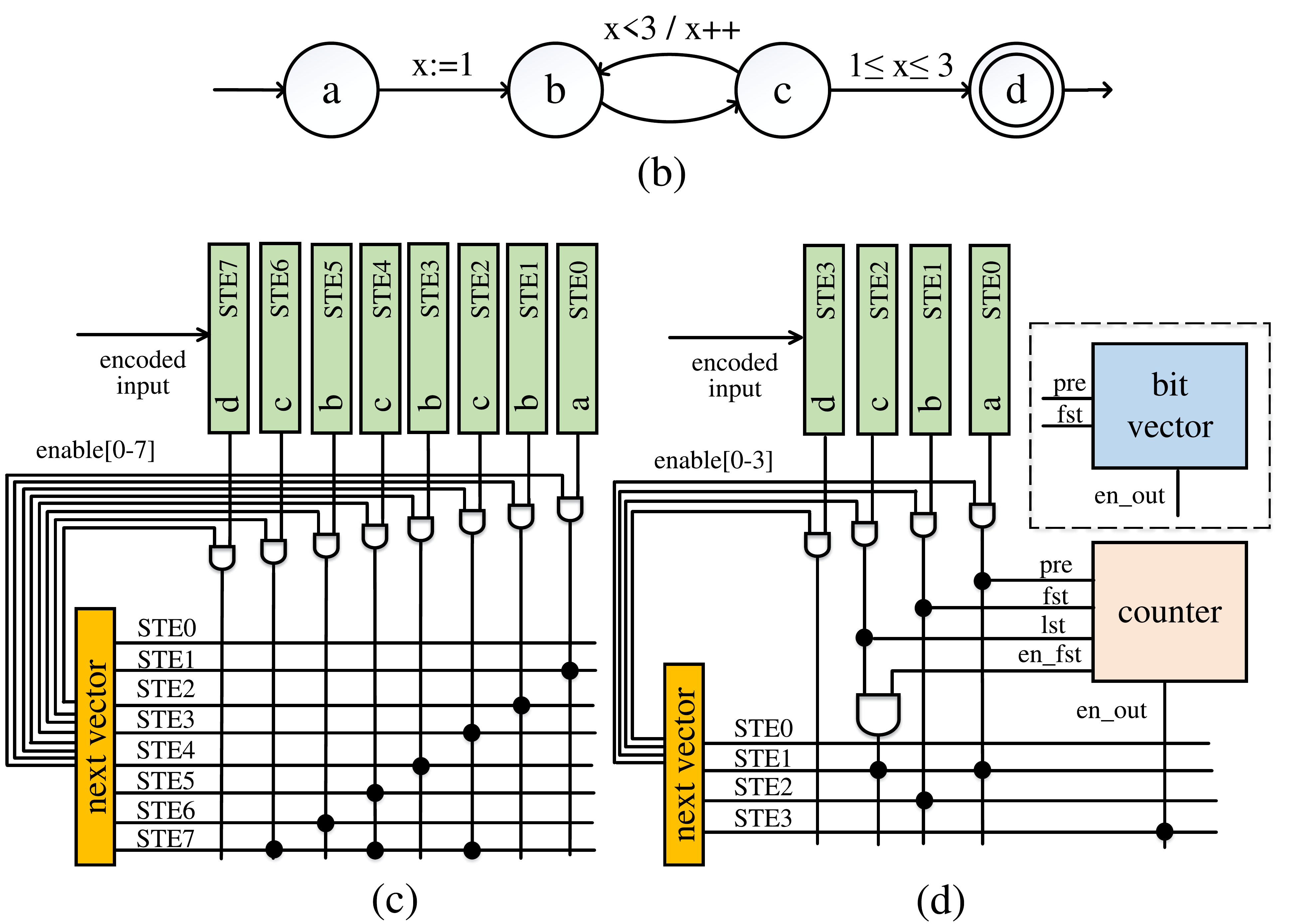}
\end{subfigure}
\caption{(a) Glushkov NCA for regex $a(bc)\{1,3\}d$. (b) Corresponding NCA with STEs. (c) Original hardware using unfolding. (d) Augmented hardware with counter or bit vector.}
\label{fig:hw-circuit}
\end{figure}
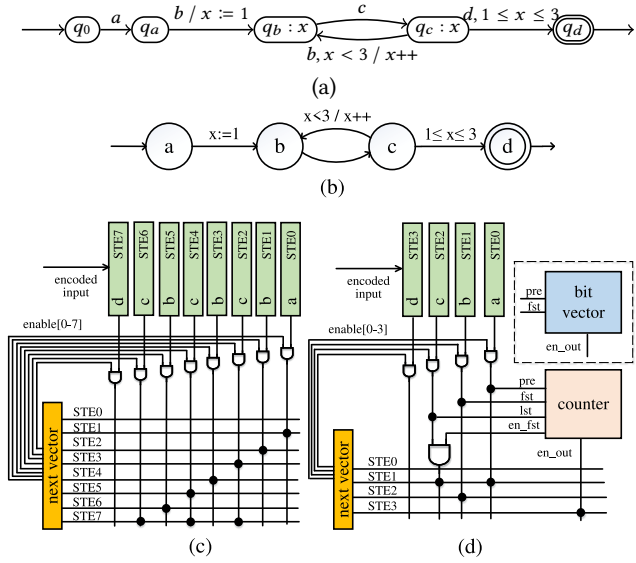

\begin{figure}
\includegraphics[width=\linewidth]{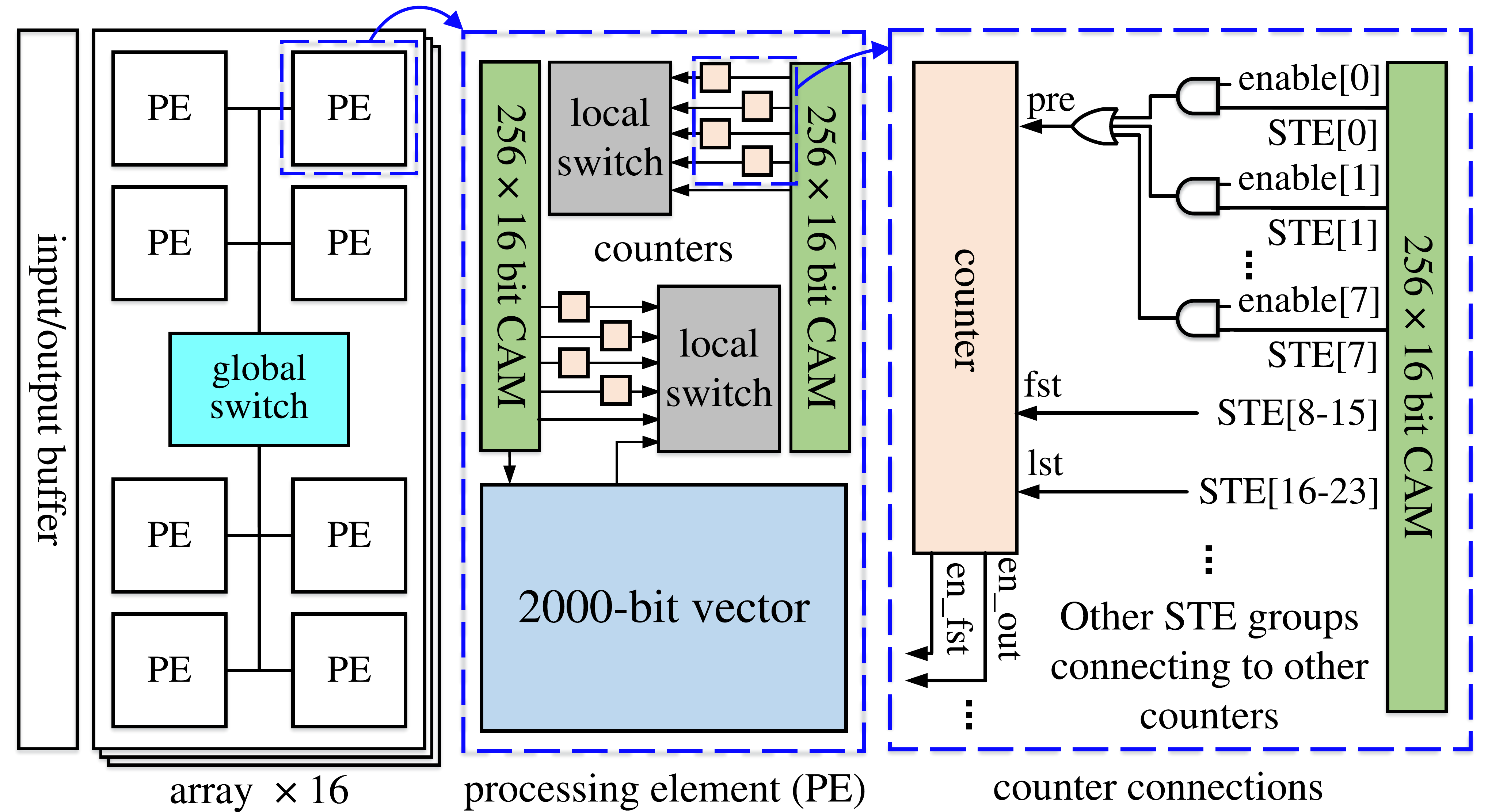}
\caption{Abstraction of proposed augmented CAMA bank, where PE is abbreviation for Processing Element.}
\label{fig:hw-arch}
\end{figure}

\paragraph{Software-Hardware Codesign}
The initial motivation for our hardware design came from the observation that several instances of bounded repetition require significantly less memory than what is suggested by a naive unfolding. This led to the formalization of counter-(un)ambiguity in NCAs and the corresponding static analysis. For the counter-unambiguous case, it suffices to use simple counter modules that keep track of the number of repetitions. For the counter-ambiguous case, the use of bit vectors is a very natural choice for a hardware representation of sets of tokens. These considerations led to the design of the counter and bit vector modules. Physical constraints imposed by the hardware call for minimizing the connections between STEs and the counting modules. For this reason, we have chosen to use bit vectors for counter-ambiguous repetitions of the form $\sigma\{m,n\}$ and use (partial) unfolding for other cases. The vast majority of counter-ambiguous repetitions in real-world benchmarks are of this form, so this approach offers efficiency (due to an optimized hardware implementation) without sacrificing generality (since the remaining cases can be handled at the level of the software/compiler).

\subsection{Compilation from Regex to MNRL}
\label{sec:toMNRL}

To program the hardware, we provide a description of the automata in the MNRL language \cite{mnrl}.
Our compiler takes a source regex and produces the MNRL file with the following steps:
(1) First, the compiler parses the regex and simplifies it with certain rewrite rules, including the unfolding of repetitions with upper bound $<2$ and the merging of character classes inside simple alternations (e.g., \texttt{[a]|[b]} is rewritten to \texttt{[ab]}).
(2) Then, the compiler performs the static analysis of \S\ref{sec:static_analysis} and annotates the regex with the counter-(un)ambiguity result for each occurrence of repetition.
(3) Finally, the compiler generates the MNRL file using these annotations, distinguishing cases where a counter suffices (counter-unambiguous) from cases where a bit vector is necessary (counter-ambiguous).

MNRL provides an element called \texttt{upCounter} for representing simple counters \cite{DlugoschBGLN2014AP, mnrl}.
However, there is no distinction between counter-ambiguous and counter-unambiguous repetition.
We have therefore extended the MNRL format by adding syntax for counters and bit vectors.

\begin{figure}
\includegraphics[width=.85\linewidth]{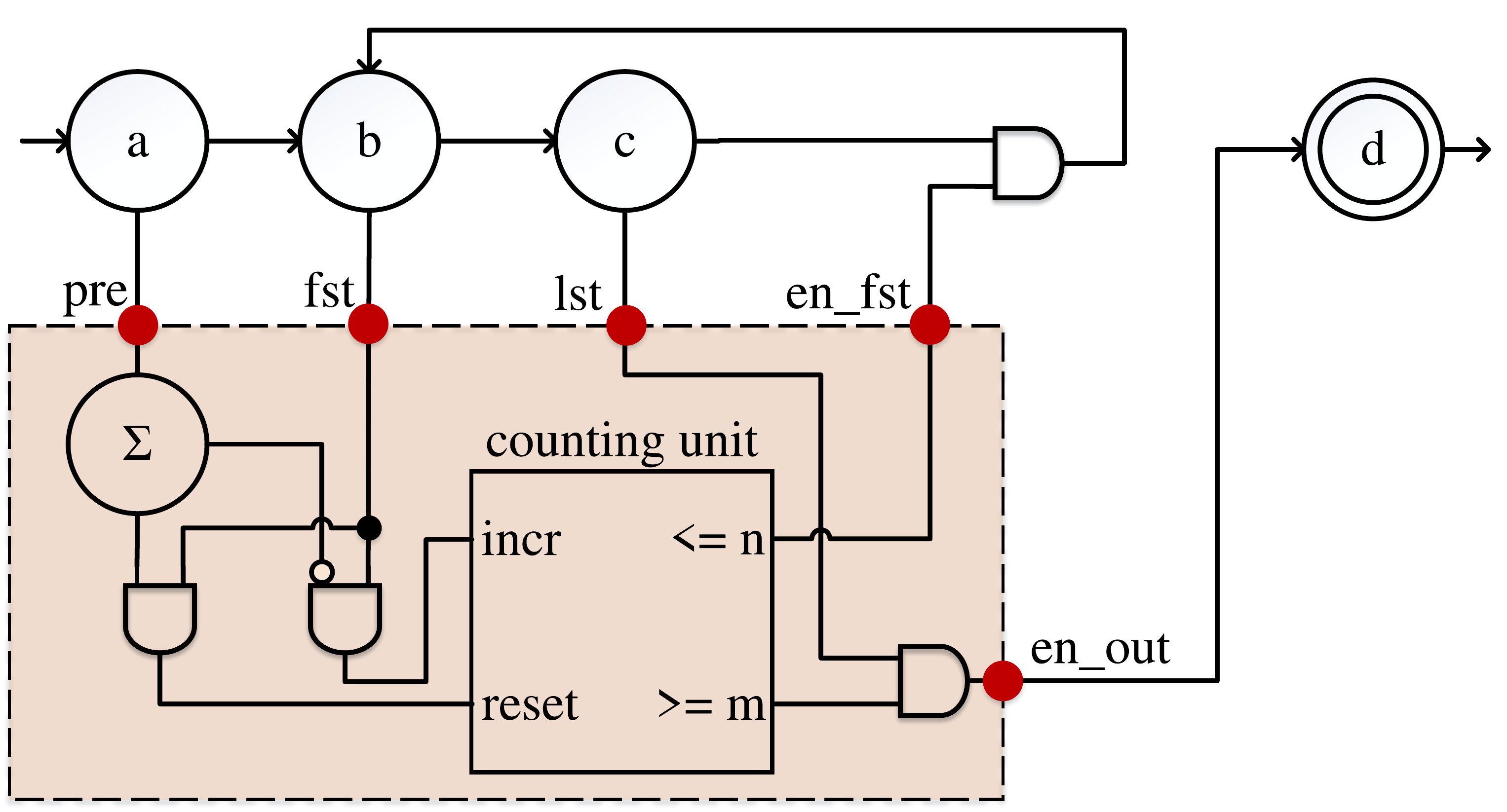}
\caption{Use of counter module to implement $a(bc)\{m,n\}d$.}
\label{fig:counter}
\end{figure}

Figure~\ref{fig:counter} presents an abstraction of the \textbf{\em counter module} (enclosed by a dashed line) by showing how it is used to implement the counter-unambiguous regex $a(bc)\{m,n\}d$ in hardware.
A counter has three incoming ports \texttt{pre}, \texttt{fst}, and \texttt{lst}, and two outgoing ports \texttt{en\_fst} and \texttt{en\_out}, where ports are labeled with red dots in Figure~\ref{fig:counter}.
The input port \texttt{pre} (i.e., pre-counting) is connected to the STE (labeled with $a$) located right before the repetition, \texttt{fst} (i.e., first) is connected to the first STE (labeled with $b$) in the repetition, and \texttt{lst} (i.e., last) is linked to the last STE (labeled with $c$) in the repetition.
The output port \texttt{en\_out} (i.e., enable output STE) activates the STE (labeled with $d$) located right after the repetition, and \texttt{en\_fst} (i.e., enable first STE) activates the first STE (labeled with $b$) in the repetition.
The counter module consists of a synchronous counting unit using D flip-flop and two digital comparators. The module is designed to meet four constraints:
(1) The counter value is reset to 0 when \texttt{pre} was active in the previous cycle and \texttt{fst} is currently active. This corresponds to the initialization of the repetition.
(2) The counter value is incremented by 1 when \texttt{fst} is active but \texttt{pre} was not active in the previous cycle. This corresponds to one complete cycle.
(3) \texttt{en\_out} fires if \texttt{lst} is active and the counter value is within the expected range (i.e., $[m,n]$).
(4) \texttt{en\_fst} fires if \texttt{lst} is active and the counter value is $\leq n$.

\begin{figure}
\includegraphics[width=.8\linewidth]{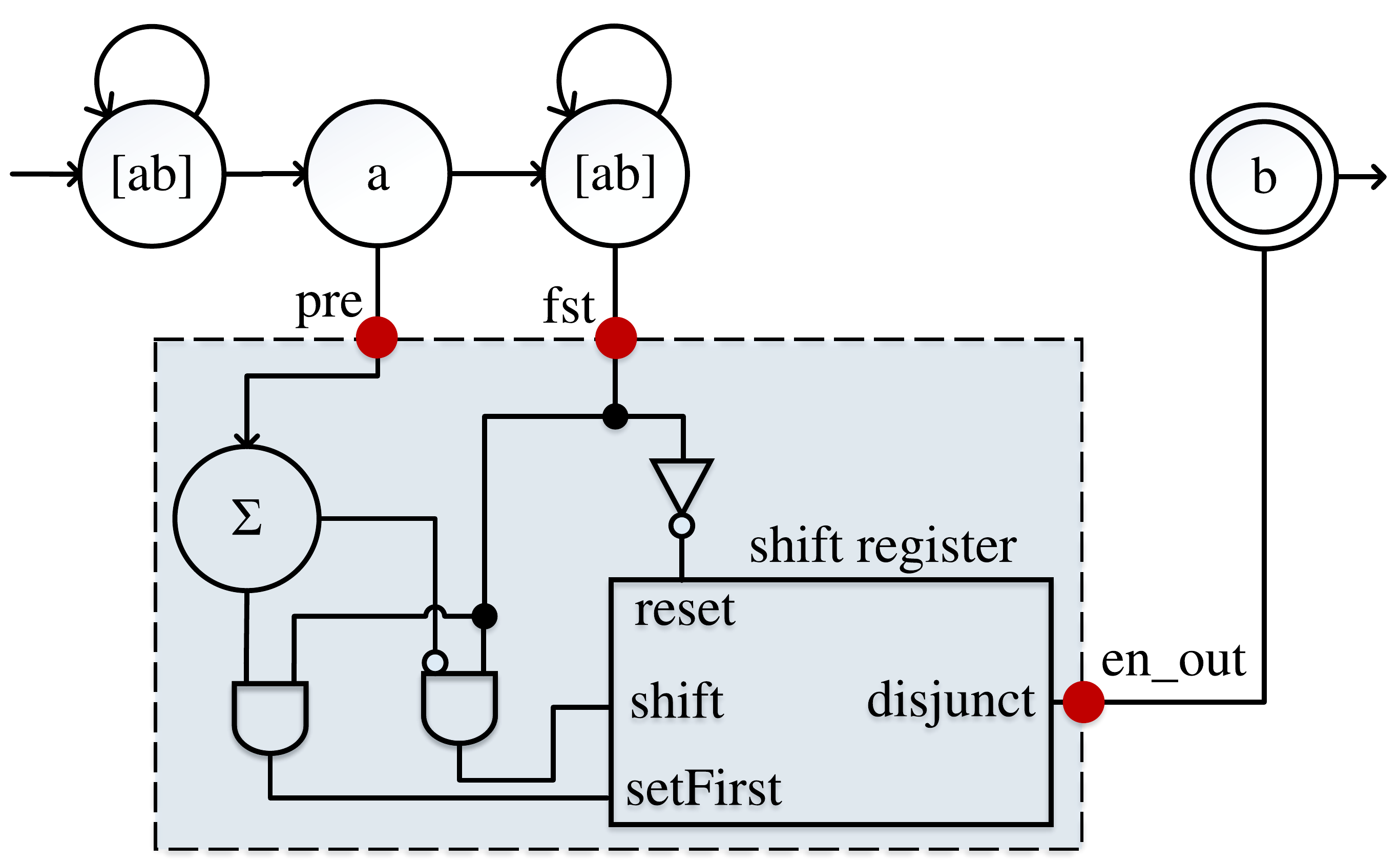}
\caption{Use of bit vector to implement $[ab]\kstar a [ab]\{m,n\} b$.}
\label{fig:bit-vector}
\end{figure}

Figure \ref{fig:bit-vector} presents an abstraction of the \textbf{\em bit vector module} by showing how the regex $[ab]\kstar a[ab]\{m,n\}b$ is implemented in hardware. 
The core component of the bit vector is a serial-in-parallel-out shift register. It supports four primary operations: (1) \texttt{reset}, which resets all bits in the vector to 0, (2) \texttt{setFirst}, which sets the first bit of the vector to 1, (3) \texttt{shift}, which shifts the vector by one bit, and (4) \texttt{disjunct}, which computes the disjunction of a sub-array of bits from index $m$ to $n$ (if one of the bits in the sub-array is 1, the output signal fires).

\subsection{Hardware Evaluation}

We modified the open-source simulator VASim \cite{ANMLZoo} to simulate the hardware performance of our augmented CAMA. 
We include 17-bit counters for supporting unambiguous counting, and 2000-bit vectors for supporting ambiguous counting, where the bit vector can be broken down to segments and used separately for counting with small upper bounds.
We use a TSMC 28nm CMOS technology and the industry-standard SPICE circuit simulator \cite{SPICE} to obtain the energy, delay, and area parameters of each component (Table \ref{table:hardwarecost}).
Since state transition is the critical path in CAMA, state matching and counter/bit-vector operations can be performed within a single clock cycle in the augmented CAMA, maintaining the same clock frequency of 2.14 GHz and throughput as CAMA-T (CAMA version optimized for high throughput) without performance penalties.

\begin{table}
\caption{Hardware component parameters}
\label{table:hardwarecost}
\vskip -2ex
\footnotesize
\centering
\renewcommand{\arraystretch}{1.1}
\begin{tabular}{|crrr|} 
\hline
Component &  \multicolumn{1}{c}{Energy (fJ)} &  \multicolumn{1}{c}{Delay (ps)} &  \multicolumn{1}{c|}{Area ($\mu m^2$)}
\\ \hline
CAMA Bank & 16780 & 325 & 3919 \\
17-bit counter & 288 & 101 & 237 \\ 
2000-bit vector & 3340 & 71 & 6382 \\
\hline
\end{tabular}
\end{table}

\begin{figure}
\includegraphics[width=.8\linewidth]{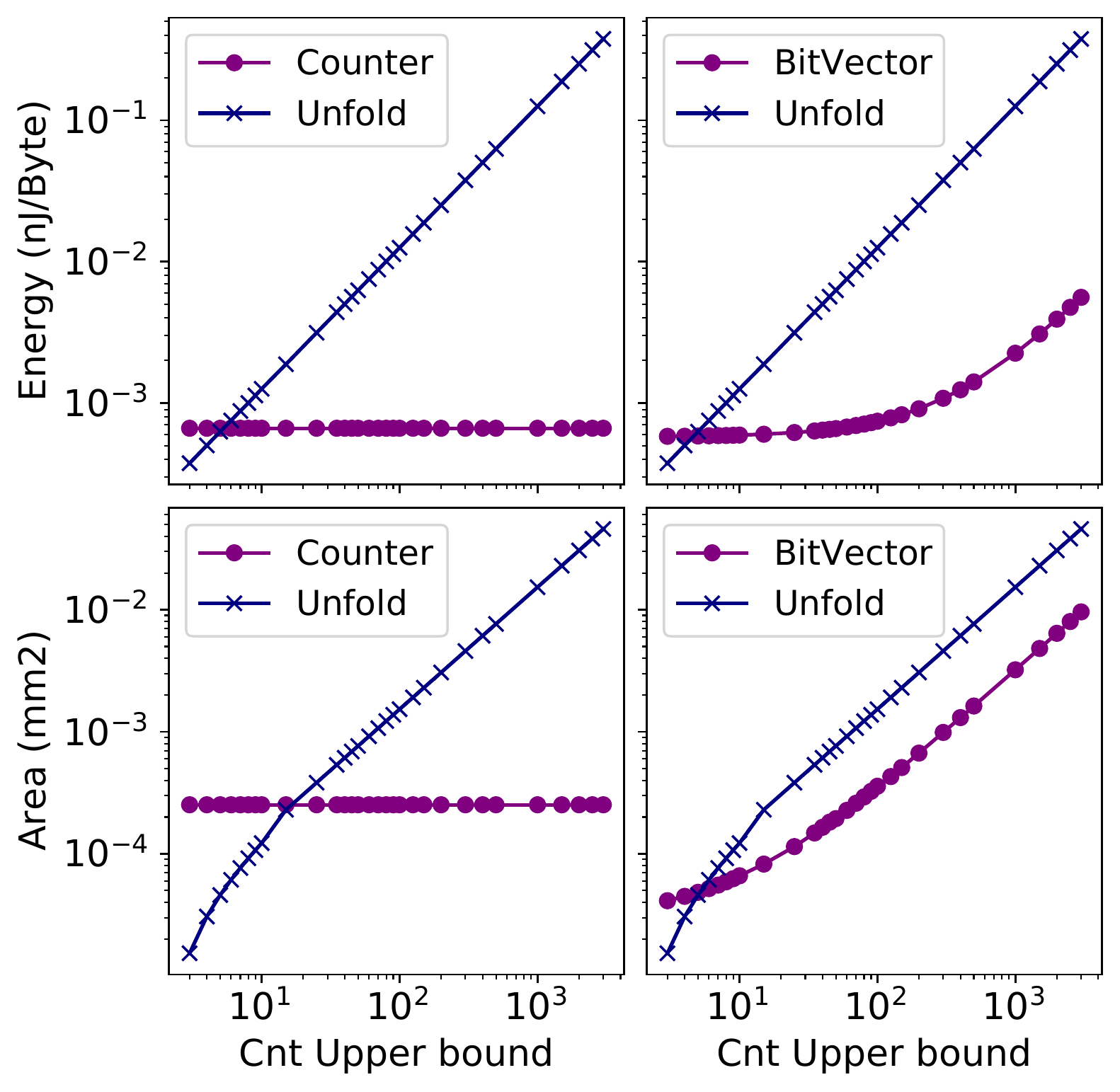}
\caption{Energy (upper two figures) and area (bottom two) trade-off of unfolding vs using counter (left two figures) and bit vector (right two), where axis is log-scaled.}
\label{fig:microbench}
\end{figure}

\begin{figure}
\includegraphics[width=.85\linewidth]{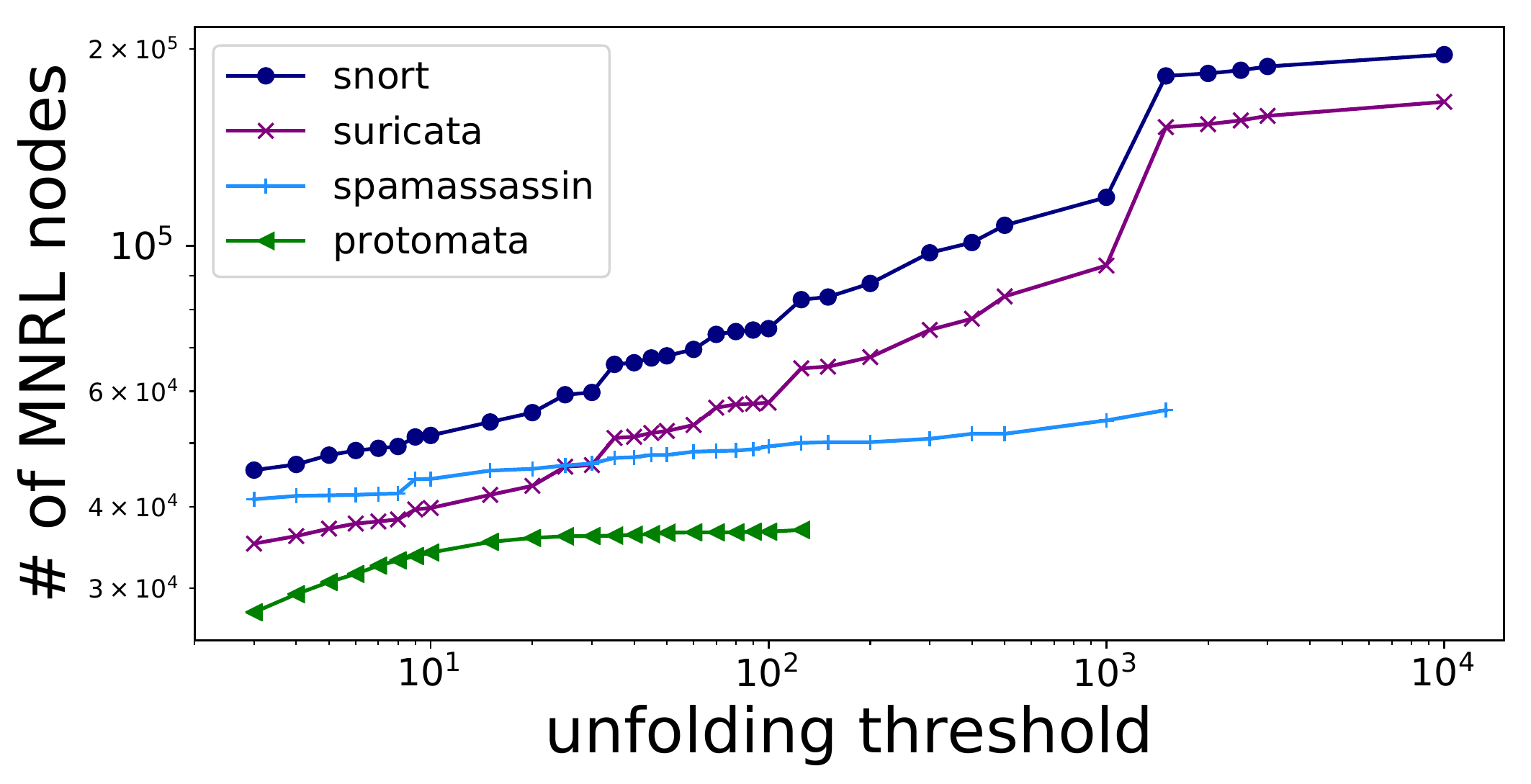}
\caption{Total number of MNRL nodes with different unfolding thresholds (both axes are log-scaled).}
\label{fig:exp-mnrl-nodes-cnt}
\end{figure}

\begin{figure*}
\includegraphics[width=.8\textwidth]{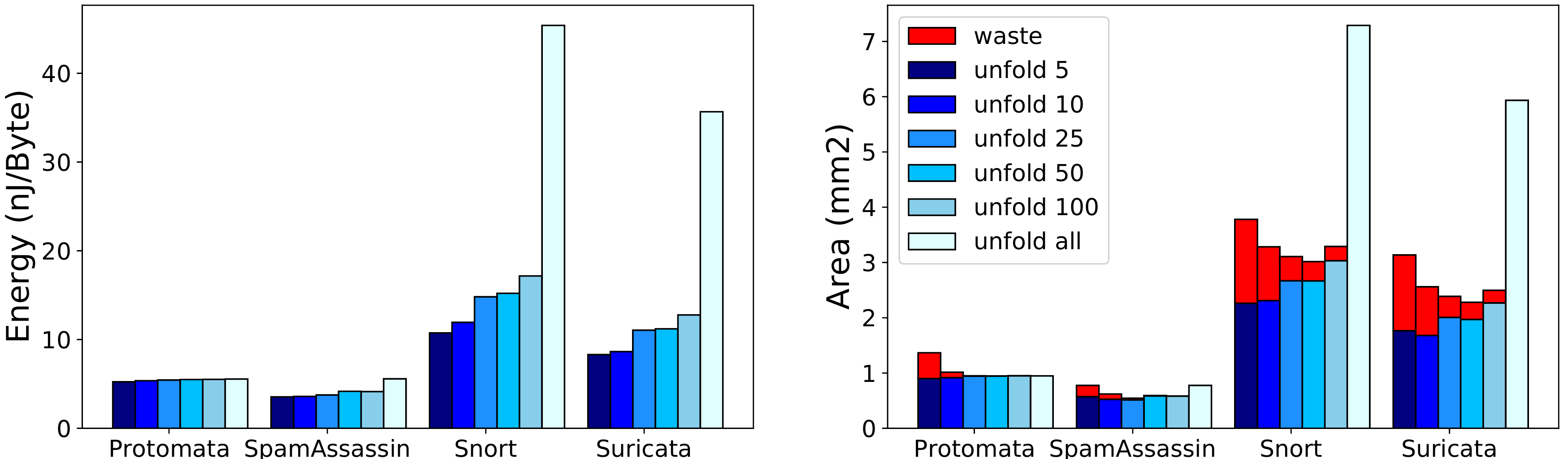}
\caption{Per-input-byte energy consumption (left) and total area cost (right) of the augmented CAMA hardware}
\label{fig:exp-mnrl-energy-memory}
\end{figure*}

\paragraph{Micro-benchmarks.}

Figure~\ref{fig:microbench} shows the trade-off of unfolding vs. using counter and bit vector modules. 
In the left two sub-figures, we consider regexes $a\{n\}$ with different values of $n$.
These regexes are counter-unambiguous -- the hardware implementation only needs a single counter module to perform the matching, while unfolding creates $n$ STEs.
The upper-left (resp., bottom-left) sub-figure shows the energy (resp., area) cost of using a counter module compared with unfolding, where we always use a 17-bit counter module to represent counter values regardless of their different repetition bounds.
In the right two sub-figures, we consider regexes $\Sigma\kstar a\{n\}$. 
These regexes are counter-ambiguous, so the hardware needs to use a bit vector to perform matching, while unfolding creates $n$ STEs.
In this comparison, we set the length of the bit vector to be equal to $n$ for each data point (this implies that bits are wasted).
The upper-right (resp., bottom-right) sub-figure shows the energy (resp., area) cost of using a bit vector compared with unfolding.
From the results shown in Figure~\ref{fig:microbench}, we observe that using a counter/bit vector provides better performance compared to unfolding even for repetitions with small upper bounds. It consistently reduces energy usage by orders of magnitude and areas by large margins.

\paragraph{Application benchmarks.}

We use the same benchmarks as described in Section~\ref{sec:static_analysis_performance} (except for ClamAV). Figure~\ref{fig:exp-mnrl-nodes-cnt} shows the number of MNRL nodes (which is linear in the number of STEs) for different unfolding thresholds. For each benchmark and each point in the corresponding curve, the x coordinate is an unfolding threshold $k$ and the y coordinate is the number of MNRL nodes that are obtained from compiling the entire benchmark after bounded repetitions up to $k$ have been unfolded. The rightmost point on each benchmark curve shows the unfolding threshold that results in full unfolding for all regexes of the benchmark and the resulting number of MNRL nodes.

We have simulated the area and the energy consumption of our augmented CAMA by feeding compiled MNRL files with different unfolding thresholds to the modified VASim.
Figure~\ref{fig:exp-mnrl-energy-memory} shows the per-input-byte energy consumption and the total area cost of the augmented CAMA.
The results show up to 76\% energy reduction and 58\% area reduction in benchmarks with an abundance of instances of bounded repetition with large upper bounds (i.e., Snort and Suricata). In benchmarks that generally include bounded repetitions with small upper bounds (i.e., Protomata and SpamAssassin), the augmented CAMA hardware still outperforms pure CAMA with little to no overhead.
We observe that for the Protomata and SpamAssassin benchmarks, our hardware implementation provides less energy and area reduction compared with Snort and Suricata. This is because, in general, the regexes in Protomata and SpamAssassin have small repetition upper bounds. The wasted area in Figure~\ref{fig:exp-mnrl-energy-memory} corresponds to unused bits in the bit vector modules.

\section{Related Work}
\label{sec:related_work}

There is a rich set of prior works that define \textbf{\em (un)ambiguity on regular expressions}. 
Book et al. \cite{BookEGO71} have defined unambiguous regexes using Glushkov automata \cite{Glushkov1961Abstract}.
Bruggemann-Klein and Wood have expressed the related notions of \emph{deterministic} \cite{Bruggemann-KleinW92} and \emph{1-unambiguous} \cite{Bruggemann-KleinW98} regexes. 
Hovland \cite{Hovland09} has defined the class of \emph{counter-1-unambiguous} for regexes with counting.
Hovland et al. \cite{HovlandDM12} have further considered a \emph{strongly 1-unambiguous} class where the membership problem, for regexes with counting and unordered concatenations, can be solved in polynomial time. Gelade et al. \cite{GeladeGM09} have defined \emph{strong} and \emph{weak determinism} and shown that weakly deterministic regexes are exponentially more succinct than the strongly deterministic ones. A survey of unambiguity in automata theory can be found at \cite{Colcombet15}.

Several different automata models and automata-based techniques have been proposed to handle \textbf{\em the matching of regexes with counting}.
DFAs and NFAs have been extended by \cite{HolikLSTVV19} and \cite{BecchiC08} respectively by introducing counting operations and guards as an alternative to unfolding for large repetition bounds.
An implementation of a class of counter automata, proposed in \cite{TuronovaHLSVV2020}, is based on queues for representing sets of counter values.
%
%
A variety of software regex matchers, including RE2 \cite{Cox10,RE2}, Rust's Regex \cite{RustRegex}, PCRE \cite{PCRE}, SRM \cite{SaarikiviVWX2019}, and Hyperscan \cite{WangHCPLHZ2019Hyperscan} support the matching of regexes with counting.
These matchers are typically based on the execution of DFAs or NFAs. Matchers like RE2 and SRM unfold constrained repetitions when performing on-the-fly determinization or computing derivatives.

A series of \textbf{\em ASIC hardware architectures} \cite{TuckNS04, BrodieBT06} have been designed to reach high throughput for network applications relying on pattern matching algorithms. The IBM regX \cite{VanLJ2012} accelerator extends the idea of representing regexes with compressed DFAs \cite{BecchiC08, YangYP2011, NakaharaHS11}, which are hybrids between DFAs and NFAs, and its parallelized architecture improves performance on large workloads. Dlugosch et al. \cite{DlugoschBGLN2014AP} designed the Automata Processor (AP), a reconfigurable ASIC hardware based on bit-parallellism \cite{Baeza89} that simulates NFAs in parallel.
Liu et al. \cite{LIKPJ2018EfficientAP} developed SparseAP to provide support for AP to efficiently execute large-scale applications.
AP can support many regexes found in real-life applications \cite{ANMLZoo, WaddenJT18}. However, it provides restricted support for regexes with counting (when upper bounds are larger than 512 they are considered unbounded \cite{RoyIS19}).
Other major ASIC works are based on the Aho-Corasick algorithm \cite{AhoAC75} including \cite{TuckNS04}, HAWK \cite{tandonPS16}, and HARE \cite{GogteKCdAW2016HARE}. 
They compute partial matches for all possible alignments and merge them to find a global match. HARE achieves a 32Gbps throughput but has limited support for Kleene operators (which only allow single character class repetition), and it provides no support for unbounded counting.

Many prior works \cite{SidhuRP01, BakerZP04} focus on \textbf{\em FPGA and GPU hardware architectures} to take advantage of their configurability and parallelism. 
\cite{YangYJ08} and \cite{SourdisIB08} provide support for regexes with counting on FPGA hardware. 
\cite{WangHP10} extends the DFA ambiguity expressed in \cite{SmithRE08} to NFA with counters by defining the \textit{character class ambiguity}, a problem that arises when the intersection between two adjacent character class with constraint repetitions (CCR) is non-empty. 
A min-max algorithm with two counters for every CCR keeps track of all possible matches. 
Our notion of counter-ambiguity is formulated more generally, and our simulation based on bit vectors handles character class ambiguity.
Finally, there are several works that implement regex matching algorithms on GPUs \cite{CascaranoNR10, ZuYY12, VasiliadisGP09, LiuPJ2020}.

\section{Conclusion}
\label{sec:conclusion}

We have investigated hardware acceleration for regular pattern matching, where the patterns are specified by regexes with an extended syntax that involves bounded repetitions of the form $r\{m,n\}$.
We have developed a design that integrates counter and bit vector modules into an in-memory NFA-based hardware architecture.
This design is inspired from the theoretical model of nondeterministic counter automata (NCAs) and the observation that some instances of bounded repetitions require only a small amount of memory. We formalize this idea using the notion of counter-unambiguity.
We have implemented a regex-to-hardware compiler that performs a static analysis for counter-(un)ambiguity over a regex and then creates a representation of an automaton with counters and bit vectors that can be deployed on the hardware.
Our experiments show that using counters and bit vectors outperforms unfolding solutions by orders of magnitude. Moreover, in experiments with realistic workloads, we have observed that our design can provide up to 76\% energy reduction and 58\% area reduction in comparison to CAMA \cite{cama}, a state-of-the-art in-memory NFA processor.

\begin{acks}                            
  We would like to thank the anonymous reviewers for their constructive comments. This research was supported in part by the US National Science Foundation award CCF 2008096 and the Rice University Faculty Initiative Fund.

\end{acks}

\bibliography{refs}


\begin{thebibliography}{70}


\ifx \showCODEN    \undefined \def \showCODEN     #1{\unskip}     \fi
\ifx \showDOI      \undefined \def \showDOI       #1{#1}\fi
\ifx \showISBNx    \undefined \def \showISBNx     #1{\unskip}     \fi
\ifx \showISBNxiii \undefined \def \showISBNxiii  #1{\unskip}     \fi
\ifx \showISSN     \undefined \def \showISSN      #1{\unskip}     \fi
\ifx \showLCCN     \undefined \def \showLCCN      #1{\unskip}     \fi
\ifx \shownote     \undefined \def \shownote      #1{#1}          \fi
\ifx \showarticletitle \undefined \def \showarticletitle #1{#1}   \fi
\ifx \showURL      \undefined \def \showURL       {\relax}        \fi
\providecommand\bibfield[2]{#2}
\providecommand\bibinfo[2]{#2}
\providecommand\natexlab[1]{#1}
\providecommand\showeprint[2][]{arXiv:#2}

\bibitem[\protect\citeauthoryear{Aho and Corasick}{Aho and Corasick}{1975}]%
        {AhoAC75}
\bibfield{author}{\bibinfo{person}{Alfred~V. Aho} {and}
  \bibinfo{person}{Margaret~J. Corasick}.} \bibinfo{year}{1975}\natexlab{}.
\newblock \showarticletitle{Efficient String Matching: An Aid to Bibliographic
  Search}.
\newblock \bibinfo{journal}{\emph{Commun. ACM}} \bibinfo{volume}{18},
  \bibinfo{number}{6} (\bibinfo{date}{jun} \bibinfo{year}{1975}),
  \bibinfo{pages}{333–340}.
\newblock
\showISSN{0001-0782}
\urldef\tempurl%
\url{https://doi.org/10.1145/360825.360855}
\showDOI{\tempurl}


\bibitem[\protect\citeauthoryear{Angstadt, Wadden, Weimer, and
  Skadron}{Angstadt et~al\mbox{.}}{2017}]%
        {mnrl}
\bibfield{author}{\bibinfo{person}{Kevin Angstadt}, \bibinfo{person}{Jack
  Wadden}, \bibinfo{person}{Westley Weimer}, {and} \bibinfo{person}{Kevin
  Skadron}.} \bibinfo{year}{2017}\natexlab{}.
\newblock \bibinfo{booktitle}{\emph{{MNRL} and {MNCaRT}: An Open-Source,
  Multi-Architecture State Machine Research and Execution Ecosystem}}.
\newblock \bibinfo{type}{{T}echnical {R}eport} CS2017-01.
  \bibinfo{institution}{University of Virginia}.
\newblock
\urldef\tempurl%
\url{https://doi.org/10.18130/V3FN18}
\showDOI{\tempurl}


\bibitem[\protect\citeauthoryear{Apache SpamAssassin}{Apache
  SpamAssassin}{2022}]%
        {SpamAssassin}
Apache SpamAssassin \bibinfo{year}{2022}\natexlab{}.
\newblock \bibinfo{title}{Apache SpamAssassin: An Open Source Anti-spam
  Platform.}
\newblock \bibinfo{howpublished}{\url{http://spamassassin.apache.org/}}.
\newblock


\bibitem[\protect\citeauthoryear{Baeza-Yates and Gonnet}{Baeza-Yates and
  Gonnet}{1989}]%
        {Baeza89}
\bibfield{author}{\bibinfo{person}{Ricardo~A. Baeza-Yates} {and}
  \bibinfo{person}{Gaston~H. Gonnet}.} \bibinfo{year}{1989}\natexlab{}.
\newblock \showarticletitle{Efficient Text Searching of Regular Expressions}.
  In \bibinfo{booktitle}{\emph{Automata, Languages and Programming}},
  \bibfield{editor}{\bibinfo{person}{Giorgio Ausiello},
  \bibinfo{person}{Mariangiola Dezani-Ciancaglini}, {and}
  \bibinfo{person}{Simonetta~Ronchi Della~Rocca}} (Eds.).
  \bibinfo{publisher}{Springer}, \bibinfo{address}{Heidelberg},
  \bibinfo{pages}{46--62}.
\newblock
\urldef\tempurl%
\url{https://doi.org/10.1007/BFb0035751}
\showDOI{\tempurl}


\bibitem[\protect\citeauthoryear{Baker and Prasanna}{Baker and
  Prasanna}{2004}]%
        {BakerZP04}
\bibfield{author}{\bibinfo{person}{Zachary~K. Baker} {and}
  \bibinfo{person}{Viktor~K. Prasanna}.} \bibinfo{year}{2004}\natexlab{}.
\newblock \showarticletitle{Time and Area Efficient Pattern Matching on FPGAs}.
  In \bibinfo{booktitle}{\emph{Proceedings of the 2004 ACM/SIGDA 12th
  International Symposium on Field Programmable Gate Arrays}}
  \emph{(\bibinfo{series}{FPGA '04})}. \bibinfo{publisher}{ACM},
  \bibinfo{address}{New York, NY, USA}, \bibinfo{pages}{223–232}.
\newblock
\urldef\tempurl%
\url{https://doi.org/10.1145/968280.968312}
\showDOI{\tempurl}


\bibitem[\protect\citeauthoryear{Barringer, Goldberg, Havelund, and
  Sen}{Barringer et~al\mbox{.}}{2004}]%
        {BarringerGHS2004}
\bibfield{author}{\bibinfo{person}{Howard Barringer}, \bibinfo{person}{Allen
  Goldberg}, \bibinfo{person}{Klaus Havelund}, {and} \bibinfo{person}{Koushik
  Sen}.} \bibinfo{year}{2004}\natexlab{}.
\newblock \showarticletitle{Rule-Based Runtime Verification}. In
  \bibinfo{booktitle}{\emph{Verification, Model Checking, and Abstract
  Interpretation (VMCAI)}} \emph{(\bibinfo{series}{LNCS},
  Vol.~\bibinfo{volume}{2937})}. \bibinfo{publisher}{Springer},
  \bibinfo{address}{Heidelberg}, \bibinfo{pages}{44--57}.
\newblock
\urldef\tempurl%
\url{https://doi.org/10.1007/978-3-540-24622-0_5}
\showDOI{\tempurl}


\bibitem[\protect\citeauthoryear{Bartocci, Deshmukh, Donz{\'e}, Fainekos,
  Maler, Ni{\v{c}}kovi{\'{c}}, and Sankaranarayanan}{Bartocci
  et~al\mbox{.}}{2018}]%
        {BartocciDDFMNS2018}
\bibfield{author}{\bibinfo{person}{Ezio Bartocci}, \bibinfo{person}{Jyotirmoy
  Deshmukh}, \bibinfo{person}{Alexandre Donz{\'e}}, \bibinfo{person}{Georgios
  Fainekos}, \bibinfo{person}{Oded Maler}, \bibinfo{person}{Dejan
  Ni{\v{c}}kovi{\'{c}}}, {and} \bibinfo{person}{Sriram Sankaranarayanan}.}
  \bibinfo{year}{2018}\natexlab{}.
\newblock \showarticletitle{Specification-Based Monitoring of Cyber-Physical
  Systems: A Survey on Theory, Tools and Applications}.
\newblock In \bibinfo{booktitle}{\emph{Lectures on Runtime Verification:
  Introductory and Advanced Topics}}, \bibfield{editor}{\bibinfo{person}{Ezio
  Bartocci} {and} \bibinfo{person}{Yli{\`e}s Falcone}} (Eds.).
  \bibinfo{series}{LNCS}, Vol.~\bibinfo{volume}{10457}.
  \bibinfo{publisher}{Springer}, \bibinfo{address}{Cham},
  \bibinfo{pages}{135--175}.
\newblock
\urldef\tempurl%
\url{https://doi.org/10.1007/978-3-319-75632-5_5}
\showDOI{\tempurl}


\bibitem[\protect\citeauthoryear{Becchi and Crowley}{Becchi and
  Crowley}{2008}]%
        {BecchiC08}
\bibfield{author}{\bibinfo{person}{Michela Becchi} {and}
  \bibinfo{person}{Patrick Crowley}.} \bibinfo{year}{2008}\natexlab{}.
\newblock \showarticletitle{Extending Finite Automata to Efficiently Match
  Perl-Compatible Regular Expressions}. In
  \bibinfo{booktitle}{\emph{Proceedings of the 2008 ACM CoNEXT Conference}}
  \emph{(\bibinfo{series}{CoNEXT '08})}. \bibinfo{publisher}{ACM},
  \bibinfo{address}{New York, NY, USA}, Article \bibinfo{articleno}{25},
  \bibinfo{numpages}{12}~pages.
\newblock
\urldef\tempurl%
\url{https://doi.org/10.1145/1544012.1544037}
\showDOI{\tempurl}


\bibitem[\protect\citeauthoryear{Bo, Dang, Sadredini, and Skadron}{Bo
  et~al\mbox{.}}{2018}]%
        {BoDSS2018}
\bibfield{author}{\bibinfo{person}{Chunkun Bo}, \bibinfo{person}{Vinh Dang},
  \bibinfo{person}{Elaheh Sadredini}, {and} \bibinfo{person}{Kevin Skadron}.}
  \bibinfo{year}{2018}\natexlab{}.
\newblock \showarticletitle{Searching for Potential {gRNA} Off-Target Sites for
  {CRISPR/Cas9} Using Automata Processing Across Different Platforms}. In
  \bibinfo{booktitle}{\emph{2018 IEEE International Symposium on High
  Performance Computer Architecture (HPCA)}}. \bibinfo{publisher}{IEEE},
  \bibinfo{pages}{737--748}.
\newblock
\urldef\tempurl%
\url{https://doi.org/10.1109/HPCA.2018.00068}
\showDOI{\tempurl}


\bibitem[\protect\citeauthoryear{Book, Even, Greibach, and Ott}{Book
  et~al\mbox{.}}{1971}]%
        {BookEGO71}
\bibfield{author}{\bibinfo{person}{Ronald Book}, \bibinfo{person}{Shimon Even},
  \bibinfo{person}{Sheila Greibach}, {and} \bibinfo{person}{Gene Ott}.}
  \bibinfo{year}{1971}\natexlab{}.
\newblock \showarticletitle{Ambiguity in Graphs and Expressions}.
\newblock \bibinfo{journal}{\emph{{IEEE} Trans. Computers}}
  \bibinfo{volume}{C-20}, \bibinfo{number}{2} (\bibinfo{year}{1971}),
  \bibinfo{pages}{149--153}.
\newblock
\urldef\tempurl%
\url{https://doi.org/10.1109/T-C.1971.223204}
\showDOI{\tempurl}


\bibitem[\protect\citeauthoryear{Brodie, Taylor, and Cytron}{Brodie
  et~al\mbox{.}}{2006}]%
        {BrodieBT06}
\bibfield{author}{\bibinfo{person}{Benjamin~C Brodie}, \bibinfo{person}{David~E
  Taylor}, {and} \bibinfo{person}{Ron~K Cytron}.}
  \bibinfo{year}{2006}\natexlab{}.
\newblock \showarticletitle{A Scalable Architecture for High-throughput
  Regular-expression Pattern Matching}.
\newblock \bibinfo{journal}{\emph{ACM SIGARCH computer architecture news}}
  \bibinfo{volume}{34}, \bibinfo{number}{2} (\bibinfo{year}{2006}),
  \bibinfo{pages}{191--202}.
\newblock


\bibitem[\protect\citeauthoryear{Br{\"u}ggemann-Klein and
  Wood}{Br{\"u}ggemann-Klein and Wood}{1992}]%
        {Bruggemann-KleinW92}
\bibfield{author}{\bibinfo{person}{Anne Br{\"u}ggemann-Klein} {and}
  \bibinfo{person}{Derick Wood}.} \bibinfo{year}{1992}\natexlab{}.
\newblock \showarticletitle{Deterministic Regular Languages}. In
  \bibinfo{booktitle}{\emph{STACS 92}}. \bibinfo{publisher}{Springer},
  \bibinfo{address}{Heidelberg}, \bibinfo{pages}{173--184}.
\newblock


\bibitem[\protect\citeauthoryear{Br{\"{u}}ggemann{-}Klein and
  Wood}{Br{\"{u}}ggemann{-}Klein and Wood}{1998}]%
        {Bruggemann-KleinW98}
\bibfield{author}{\bibinfo{person}{Anne Br{\"{u}}ggemann{-}Klein} {and}
  \bibinfo{person}{Derick Wood}.} \bibinfo{year}{1998}\natexlab{}.
\newblock \showarticletitle{One-Unambiguous Regular Languages}.
\newblock \bibinfo{journal}{\emph{Inf. Comput.}} \bibinfo{volume}{140},
  \bibinfo{number}{2} (\bibinfo{year}{1998}), \bibinfo{pages}{229--253}.
\newblock
\urldef\tempurl%
\url{https://doi.org/10.1006/inco.1997.2688}
\showDOI{\tempurl}


\bibitem[\protect\citeauthoryear{Cascarano, Rolando, Risso, and
  Sisto}{Cascarano et~al\mbox{.}}{2010}]%
        {CascaranoNR10}
\bibfield{author}{\bibinfo{person}{Niccolo' Cascarano},
  \bibinfo{person}{Pierluigi Rolando}, \bibinfo{person}{Fulvio Risso}, {and}
  \bibinfo{person}{Riccardo Sisto}.} \bibinfo{year}{2010}\natexlab{}.
\newblock \showarticletitle{iNFAnt: NFA Pattern Matching on GPGPU Devices}.
\newblock \bibinfo{journal}{\emph{ACM SIGCOMM Computer Communication Review}}
  \bibinfo{volume}{40}, \bibinfo{number}{5} (\bibinfo{year}{2010}),
  \bibinfo{pages}{20--26}.
\newblock


\bibitem[\protect\citeauthoryear{Chattopadhyay and Mamouras}{Chattopadhyay and
  Mamouras}{2020}]%
        {ChattopadhyayM2020}
\bibfield{author}{\bibinfo{person}{Agnishom Chattopadhyay} {and}
  \bibinfo{person}{Konstantinos Mamouras}.} \bibinfo{year}{2020}\natexlab{}.
\newblock \showarticletitle{A Verified Online Monitor for Metric Temporal Logic
  with Quantitative Semantics}. In \bibinfo{booktitle}{\emph{Runtime
  Verification (RV)}} \emph{(\bibinfo{series}{LNCS},
  Vol.~\bibinfo{volume}{12399})}, \bibfield{editor}{\bibinfo{person}{Jyotirmoy
  Deshmukh} {and} \bibinfo{person}{Dejan Ni{\v{c}}kovi{\'{c}}}} (Eds.).
  \bibinfo{publisher}{Springer}, \bibinfo{address}{Cham},
  \bibinfo{pages}{383--403}.
\newblock
\urldef\tempurl%
\url{https://doi.org/10.1007/978-3-030-60508-7_21}
\showDOI{\tempurl}


\bibitem[\protect\citeauthoryear{ClamAV}{ClamAV}{2022}]%
        {ClamAV}
ClamAV \bibinfo{year}{2022}\natexlab{}.
\newblock \bibinfo{title}{ClamAV\textregistered: An Open-source Antivirus
  Engine for Detecting Trojans, Viruses, Malware \& Other Malicious Threats.}
\newblock \bibinfo{howpublished}{\url{https://www.clamav.net/}}.
\newblock


\bibitem[\protect\citeauthoryear{Colcombet}{Colcombet}{2015}]%
        {Colcombet15}
\bibfield{author}{\bibinfo{person}{Thomas Colcombet}.}
  \bibinfo{year}{2015}\natexlab{}.
\newblock \showarticletitle{Unambiguity in Automata Theory}. In
  \bibinfo{booktitle}{\emph{Descriptional Complexity of Formal Systems}}.
  \bibinfo{publisher}{Springer}, \bibinfo{address}{Cham},
  \bibinfo{pages}{3--18}.
\newblock
\urldef\tempurl%
\url{https://doi.org/10.1007/978-3-319-19225-3_1}
\showDOI{\tempurl}


\bibitem[\protect\citeauthoryear{Cox}{Cox}{2010}]%
        {Cox10}
\bibfield{author}{\bibinfo{person}{Russ Cox}.} \bibinfo{year}{2010}\natexlab{}.
\newblock \bibinfo{title}{Regular Expression Matching in the Wild}.
\newblock
  \bibinfo{howpublished}{\url{https://swtch.com/~rsc/regexp/regexp3.html}}.
\newblock


\bibitem[\protect\citeauthoryear{Dlugosch, Brown, Glendenning, Leventhal, and
  Noyes}{Dlugosch et~al\mbox{.}}{2014}]%
        {DlugoschBGLN2014AP}
\bibfield{author}{\bibinfo{person}{Paul Dlugosch}, \bibinfo{person}{Dave
  Brown}, \bibinfo{person}{Paul Glendenning}, \bibinfo{person}{Michael
  Leventhal}, {and} \bibinfo{person}{Harold Noyes}.}
  \bibinfo{year}{2014}\natexlab{}.
\newblock \showarticletitle{An Efficient and Scalable Semiconductor
  Architecture for Parallel Automata Processing}.
\newblock \bibinfo{journal}{\emph{IEEE Transactions on Parallel and Distributed
  Systems}} \bibinfo{volume}{25}, \bibinfo{number}{12} (\bibinfo{year}{2014}),
  \bibinfo{pages}{3088--3098}.
\newblock
\urldef\tempurl%
\url{https://doi.org/10.1109/TPDS.2014.8}
\showDOI{\tempurl}


\bibitem[\protect\citeauthoryear{Gelade, Gyssens, and Martens}{Gelade
  et~al\mbox{.}}{2009}]%
        {GeladeGM09}
\bibfield{author}{\bibinfo{person}{Wouter Gelade}, \bibinfo{person}{Marc
  Gyssens}, {and} \bibinfo{person}{Wim Martens}.}
  \bibinfo{year}{2009}\natexlab{}.
\newblock \showarticletitle{Regular Expressions with Counting: Weak versus
  Strong Determinism}. In \bibinfo{booktitle}{\emph{Mathematical Foundations of
  Computer Science 2009}}. \bibinfo{publisher}{Springer},
  \bibinfo{address}{Heidelberg}, \bibinfo{pages}{369--381}.
\newblock


\bibitem[\protect\citeauthoryear{Glushkov}{Glushkov}{1961}]%
        {Glushkov1961Abstract}
\bibfield{author}{\bibinfo{person}{Victor~Mikhaylovich Glushkov}.}
  \bibinfo{year}{1961}\natexlab{}.
\newblock \showarticletitle{The Abstract Theory of Automata}.
\newblock \bibinfo{journal}{\emph{Russian Math. Surveys}} \bibinfo{volume}{16},
  \bibinfo{number}{5} (\bibinfo{year}{1961}), \bibinfo{pages}{1--53}.
\newblock
\urldef\tempurl%
\url{https://doi.org/10.1070/RM1961v016n05ABEH004112}
\showDOI{\tempurl}


\bibitem[\protect\citeauthoryear{Gogte, Kolli, Cafarella, D'Antoni, and
  Wenisch}{Gogte et~al\mbox{.}}{2016}]%
        {GogteKCdAW2016HARE}
\bibfield{author}{\bibinfo{person}{Vaibhav Gogte}, \bibinfo{person}{Aasheesh
  Kolli}, \bibinfo{person}{Michael~J. Cafarella}, \bibinfo{person}{Loris
  D'Antoni}, {and} \bibinfo{person}{Thomas~F. Wenisch}.}
  \bibinfo{year}{2016}\natexlab{}.
\newblock \showarticletitle{{HARE}: Hardware Accelerator for Regular
  Expressions}. In \bibinfo{booktitle}{\emph{2016 49th Annual IEEE/ACM
  International Symposium on Microarchitecture (MICRO)}}.
  \bibinfo{publisher}{IEEE}, \bibinfo{pages}{1--12}.
\newblock
\urldef\tempurl%
\url{https://doi.org/10.1109/MICRO.2016.7783747}
\showDOI{\tempurl}


\bibitem[\protect\citeauthoryear{Hol{\'i}k, Leng{\'a}l, Saarikivi,
  Turo{\v{n}}ov{\'a}, Veanes, and Vojnar}{Hol{\'i}k et~al\mbox{.}}{2019}]%
        {HolikLSTVV19}
\bibfield{author}{\bibinfo{person}{Luk{\'a}{\v{s}} Hol{\'i}k},
  \bibinfo{person}{Ond{\v{r}}ej Leng{\'a}l}, \bibinfo{person}{Olli Saarikivi},
  \bibinfo{person}{Lenka Turo{\v{n}}ov{\'a}}, \bibinfo{person}{Margus Veanes},
  {and} \bibinfo{person}{Tom{\'a}{\v{s}} Vojnar}.}
  \bibinfo{year}{2019}\natexlab{}.
\newblock \showarticletitle{Succinct Determinisation of Counting Automata via
  Sphere Construction}. In \bibinfo{booktitle}{\emph{Programming Languages and
  Systems}}, \bibfield{editor}{\bibinfo{person}{Anthony~Widjaja Lin}} (Ed.).
  \bibinfo{publisher}{Springer}, \bibinfo{address}{Cham},
  \bibinfo{pages}{468--489}.
\newblock
\urldef\tempurl%
\url{https://doi.org/10.1007/978-3-030-34175-6_24}
\showDOI{\tempurl}


\bibitem[\protect\citeauthoryear{Hovland}{Hovland}{2009}]%
        {Hovland09}
\bibfield{author}{\bibinfo{person}{Dag Hovland}.}
  \bibinfo{year}{2009}\natexlab{}.
\newblock \showarticletitle{Regular Expressions with Numerical Constraints and
  Automata with Counters}. In \bibinfo{booktitle}{\emph{Theoretical Aspects of
  Computing - ICTAC 2009}}. \bibinfo{publisher}{Springer},
  \bibinfo{address}{Heidelberg}, \bibinfo{pages}{231--245}.
\newblock
\urldef\tempurl%
\url{https://doi.org/10.1007/978-3-642-03466-4_15}
\showDOI{\tempurl}


\bibitem[\protect\citeauthoryear{Hovland}{Hovland}{2012}]%
        {HovlandDM12}
\bibfield{author}{\bibinfo{person}{Dag Hovland}.}
  \bibinfo{year}{2012}\natexlab{}.
\newblock \showarticletitle{The Membership Problem for Regular Expressions with
  Unordered Concatenation and Numerical Constraints}. In
  \bibinfo{booktitle}{\emph{Language and Automata Theory and Applications}}.
  \bibinfo{publisher}{Springer}, \bibinfo{address}{Heidelberg},
  \bibinfo{pages}{313--324}.
\newblock


\bibitem[\protect\citeauthoryear{Huang, Chen, Li, and Yang}{Huang
  et~al\mbox{.}}{2021}]%
        {cama}
\bibfield{author}{\bibinfo{person}{Yi Huang}, \bibinfo{person}{Zhiyu Chen},
  \bibinfo{person}{Dai Li}, {and} \bibinfo{person}{Kaiyuan Yang}.}
  \bibinfo{year}{2021}\natexlab{}.
\newblock \bibinfo{title}{CAMA: Energy and Memory Efficient Automata Processing
  in Content-Addressable Memories}.
\newblock
\newblock
\urldef\tempurl%
\url{https://doi.org/10.48550/arXiv.2112.00267}
\showDOI{\tempurl}
\showeprint[arxiv]{2112.00267}~[cs.AR]


\bibitem[\protect\citeauthoryear{Lenjani and Hashemi}{Lenjani and
  Hashemi}{2014}]%
        {LenjaniMH14}
\bibfield{author}{\bibinfo{person}{Marzieh Lenjani} {and}
  \bibinfo{person}{Mahmoud~Reza Hashemi}.} \bibinfo{year}{2014}\natexlab{}.
\newblock \showarticletitle{Tree-based Scheme for Reducing Shared Cache Miss
  Rate Leveraging Regional, Statistical and Temporal Similarities}.
\newblock \bibinfo{journal}{\emph{IET Computers \& Digital Techniques}}
  \bibinfo{volume}{8}, \bibinfo{number}{1} (\bibinfo{year}{2014}),
  \bibinfo{pages}{30--48}.
\newblock
\urldef\tempurl%
\url{https://doi.org/10.1049/iet-cdt.2011.0066}
\showDOI{\tempurl}


\bibitem[\protect\citeauthoryear{Liu, Ibrahim, Kayiran, Pai, and Jog}{Liu
  et~al\mbox{.}}{2018}]%
        {LIKPJ2018EfficientAP}
\bibfield{author}{\bibinfo{person}{Hongyuan Liu}, \bibinfo{person}{Mohamed
  Ibrahim}, \bibinfo{person}{Onur Kayiran}, \bibinfo{person}{Sreepathi Pai},
  {and} \bibinfo{person}{Adwait Jog}.} \bibinfo{year}{2018}\natexlab{}.
\newblock \showarticletitle{Architectural Support for Efficient Large-Scale
  Automata Processing}. In \bibinfo{booktitle}{\emph{2018 51st Annual IEEE/ACM
  International Symposium on Microarchitecture (MICRO)}}.
  \bibinfo{publisher}{IEEE}, \bibinfo{address}{New York, NY, USA},
  \bibinfo{pages}{908--920}.
\newblock
\urldef\tempurl%
\url{https://doi.org/10.1109/MICRO.2018.00078}
\showDOI{\tempurl}


\bibitem[\protect\citeauthoryear{Liu, Pai, and Jog}{Liu et~al\mbox{.}}{2020}]%
        {LiuPJ2020}
\bibfield{author}{\bibinfo{person}{Hongyuan Liu}, \bibinfo{person}{Sreepathi
  Pai}, {and} \bibinfo{person}{Adwait Jog}.} \bibinfo{year}{2020}\natexlab{}.
\newblock \showarticletitle{Why {GPUs} Are Slow at Executing {NFAs} and How to
  Make Them Faster}. In \bibinfo{booktitle}{\emph{Proceedings of the
  Twenty-Fifth International Conference on Architectural Support for
  Programming Languages and Operating Systems}} \emph{(\bibinfo{series}{ASPLOS
  '20})}. \bibinfo{publisher}{ACM}, \bibinfo{address}{New York, NY, USA},
  \bibinfo{pages}{251--265}.
\newblock
\urldef\tempurl%
\url{https://doi.org/10.1145/3373376.3378471}
\showDOI{\tempurl}


\bibitem[\protect\citeauthoryear{Liu, Yang, Liu, Sun, and Guo}{Liu
  et~al\mbox{.}}{2011}]%
        {LiuTY11}
\bibfield{author}{\bibinfo{person}{T. Liu}, \bibinfo{person}{Y. Yang},
  \bibinfo{person}{Y. Liu}, \bibinfo{person}{Y. Sun}, {and} \bibinfo{person}{Li
  Guo}.} \bibinfo{year}{2011}\natexlab{}.
\newblock \showarticletitle{An Efficient Regular Expressions Compression
  Algorithm from a New Perspective}. In \bibinfo{booktitle}{\emph{2011
  {{Proceedings IEEE INFOCOM}}}}. \bibinfo{publisher}{IEEE},
  \bibinfo{address}{New York, NY, USA}, \bibinfo{pages}{2129--2137}.
\newblock
\urldef\tempurl%
\url{https://doi.org/10.1109/INFCOM.2011.5935024}
\showDOI{\tempurl}


\bibitem[\protect\citeauthoryear{Lunteren, Hagleitner, Heil, Biran, Shvadron,
  and Atasu}{Lunteren et~al\mbox{.}}{2012}]%
        {VanLJ2012}
\bibfield{author}{\bibinfo{person}{Jan~Van Lunteren},
  \bibinfo{person}{Christoph Hagleitner}, \bibinfo{person}{Timothy Heil},
  \bibinfo{person}{Giora Biran}, \bibinfo{person}{Uzi Shvadron}, {and}
  \bibinfo{person}{Kubilay Atasu}.} \bibinfo{year}{2012}\natexlab{}.
\newblock \showarticletitle{Designing a Programmable Wire-Speed
  Regular-Expression Matching Accelerator}. In \bibinfo{booktitle}{\emph{2012
  45th Annual IEEE/ACM International Symposium on Microarchitecture}}.
  \bibinfo{publisher}{IEEE}, \bibinfo{address}{New York, NY, USA},
  \bibinfo{pages}{461--472}.
\newblock
\urldef\tempurl%
\url{https://doi.org/10.1109/MICRO.2012.49}
\showDOI{\tempurl}


\bibitem[\protect\citeauthoryear{Mamouras, Chattopadhyay, and Wang}{Mamouras
  et~al\mbox{.}}{2021}]%
        {MamourasCW2021TACAS}
\bibfield{author}{\bibinfo{person}{Konstantinos Mamouras},
  \bibinfo{person}{Agnishom Chattopadhyay}, {and} \bibinfo{person}{Zhifu
  Wang}.} \bibinfo{year}{2021}\natexlab{}.
\newblock \showarticletitle{Algebraic Quantitative Semantics for Efficient
  Online Temporal Monitoring}. In \bibinfo{booktitle}{\emph{Tools and
  Algorithms for the Construction and Analysis of Systems (TACAS)}}
  \emph{(\bibinfo{series}{LNCS}, Vol.~\bibinfo{volume}{12651})},
  \bibfield{editor}{\bibinfo{person}{Jan~Friso Groote} {and}
  \bibinfo{person}{Kim~Guldstrand Larsen}} (Eds.).
  \bibinfo{publisher}{Springer}, \bibinfo{address}{Cham},
  \bibinfo{pages}{330--348}.
\newblock
\urldef\tempurl%
\url{https://doi.org/10.1007/978-3-030-72016-2_18}
\showDOI{\tempurl}


\bibitem[\protect\citeauthoryear{Mamouras and Wang}{Mamouras and Wang}{2020}]%
        {MamourasW2020}
\bibfield{author}{\bibinfo{person}{Konstantinos Mamouras} {and}
  \bibinfo{person}{Zhifu Wang}.} \bibinfo{year}{2020}\natexlab{}.
\newblock \showarticletitle{Online Signal Monitoring with Bounded Lag}.
\newblock \bibinfo{journal}{\emph{IEEE Transactions on Computer-Aided Design of
  Integrated Circuits and Systems}} \bibinfo{volume}{39}, \bibinfo{number}{11}
  (\bibinfo{year}{2020}), \bibinfo{pages}{3868--3880}.
\newblock
\urldef\tempurl%
\url{https://doi.org/10.1109/TCAD.2020.3013053}
\showDOI{\tempurl}


\bibitem[\protect\citeauthoryear{Meyer and Fischer}{Meyer and Fischer}{1971}]%
        {MeyerF1971}
\bibfield{author}{\bibinfo{person}{Albert~R. Meyer} {and}
  \bibinfo{person}{Michael~J. Fischer}.} \bibinfo{year}{1971}\natexlab{}.
\newblock \showarticletitle{Economy of Description by Automata, Grammars, and
  Formal Systems}. In \bibinfo{booktitle}{\emph{2013 IEEE 54th Annual Symposium
  on Foundations of Computer Science}}. \bibinfo{publisher}{IEEE Computer
  Society}, \bibinfo{address}{Los Alamitos, CA, USA},
  \bibinfo{pages}{188--191}.
\newblock
\urldef\tempurl%
\url{https://doi.org/10.1109/SWAT.1971.11}
\showDOI{\tempurl}


\bibitem[\protect\citeauthoryear{Meyer and Stockmeyer}{Meyer and
  Stockmeyer}{1972}]%
        {MeyerS1972}
\bibfield{author}{\bibinfo{person}{Albert~R. Meyer} {and}
  \bibinfo{person}{Larry~J. Stockmeyer}.} \bibinfo{year}{1972}\natexlab{}.
\newblock \showarticletitle{The Equivalence Problem for Regular Expressions
  with Squaring Requires Exponential Space}. In \bibinfo{booktitle}{\emph{13th
  Annual Symposium on Switching and Automata Theory (SWAT 1972)}}.
  \bibinfo{publisher}{IEEE Computer Society}, \bibinfo{address}{Los Alamitos,
  CA, USA}, \bibinfo{pages}{125--129}.
\newblock
\urldef\tempurl%
\url{https://doi.org/10.1109/SWAT.1972.29}
\showDOI{\tempurl}


\bibitem[\protect\citeauthoryear{Nakahara, Sasao, and Matsuura}{Nakahara
  et~al\mbox{.}}{2011}]%
        {NakaharaHS11}
\bibfield{author}{\bibinfo{person}{Hiroki Nakahara}, \bibinfo{person}{Tsutomu
  Sasao}, {and} \bibinfo{person}{Munehiro Matsuura}.}
  \bibinfo{year}{2011}\natexlab{}.
\newblock \showarticletitle{A Regular Expression Matching Circuit Based on a
  Decomposed Automaton}. In \bibinfo{booktitle}{\emph{Reconfigurable Computing:
  Architectures, Tools and Applications}}. \bibinfo{publisher}{Springer},
  \bibinfo{address}{Heidelberg}, \bibinfo{pages}{16--28}.
\newblock
\urldef\tempurl%
\url{https://doi.org/10.1007/978-3-642-19475-7_4}
\showDOI{\tempurl}


\bibitem[\protect\citeauthoryear{PCRE}{PCRE}{2021}]%
        {PCRE}
PCRE \bibinfo{year}{2021}\natexlab{}.
\newblock \bibinfo{title}{PCRE - Perl Compatible Regular Expressions}.
\newblock \bibinfo{howpublished}{\url{https://www.pcre.org/}}.
\newblock


\bibitem[\protect\citeauthoryear{Posix Syntax in PCRE}{Posix Syntax in
  PCRE}{2022}]%
        {PosixSyntax}
Posix Syntax in PCRE \bibinfo{year}{2022}\natexlab{}.
\newblock \bibinfo{title}{Posix Syntax in PCRE}.
\newblock
  \bibinfo{howpublished}{\url{https://www.pcre.org/original/doc/html/pcrepattern.html}}.
\newblock


\bibitem[\protect\citeauthoryear{PROSITE}{PROSITE}{2022}]%
        {Prosite}
PROSITE \bibinfo{year}{2022}\natexlab{}.
\newblock \bibinfo{title}{PROSITE: Database of Protein Domains, Families and
  Functional Sites.}
\newblock \bibinfo{howpublished}{\url{https://prosite.expasy.org/}}.
\newblock


\bibitem[\protect\citeauthoryear{Rahimi, Sadredini, Stan, and Skadron}{Rahimi
  et~al\mbox{.}}{2020}]%
        {RahimiRS20}
\bibfield{author}{\bibinfo{person}{Reza Rahimi}, \bibinfo{person}{Elaheh
  Sadredini}, \bibinfo{person}{Mircea Stan}, {and} \bibinfo{person}{Kevin
  Skadron}.} \bibinfo{year}{2020}\natexlab{}.
\newblock \showarticletitle{Grapefruit: {{An Open}}-{{Source}},
  {{Full}}-{{Stack}}, and {{Customizable Automata Processing}} on {{FPGAs}}}.
  In \bibinfo{booktitle}{\emph{2020 {{IEEE}} 28th {{Annual International
  Symposium}} on {{Field}}-{{Programmable Custom Computing Machines}}
  ({{FCCM}})}}. \bibinfo{publisher}{IEEE}, \bibinfo{address}{New York, NY,
  USA}, \bibinfo{pages}{138--147}.
\newblock
\urldef\tempurl%
\url{https://doi.org/10.1109/FCCM48280.2020.00027}
\showDOI{\tempurl}


\bibitem[\protect\citeauthoryear{RE2}{RE2}{2021}]%
        {RE2}
RE2 \bibinfo{year}{2021}\natexlab{}.
\newblock \bibinfo{title}{RE2: Google's regular expression library}.
\newblock \bibinfo{howpublished}{\url{https://github.com/google/re2}}.
\newblock


\bibitem[\protect\citeauthoryear{Roy and Aluru}{Roy and Aluru}{2016}]%
        {RoyA2016}
\bibfield{author}{\bibinfo{person}{Indranil Roy} {and}
  \bibinfo{person}{Srinivas Aluru}.} \bibinfo{year}{2016}\natexlab{}.
\newblock \showarticletitle{Discovering Motifs in Biological Sequences Using
  the {M}icron {A}utomata {P}rocessor}.
\newblock \bibinfo{journal}{\emph{IEEE/ACM Transactions on Computational
  Biology and Bioinformatics}} \bibinfo{volume}{13}, \bibinfo{number}{1}
  (\bibinfo{year}{2016}), \bibinfo{pages}{99--111}.
\newblock
\urldef\tempurl%
\url{https://doi.org/10.1109/TCBB.2015.2430313}
\showDOI{\tempurl}


\bibitem[\protect\citeauthoryear{Roy, Srivastava, Grimm, Nourian, Becchi, and
  Aluru}{Roy et~al\mbox{.}}{2019}]%
        {RoyIS19}
\bibfield{author}{\bibinfo{person}{Indranil Roy}, \bibinfo{person}{Ankit
  Srivastava}, \bibinfo{person}{Matt Grimm}, \bibinfo{person}{Marziyeh
  Nourian}, \bibinfo{person}{Michela Becchi}, {and} \bibinfo{person}{Srinivas
  Aluru}.} \bibinfo{year}{2019}\natexlab{}.
\newblock \showarticletitle{Evaluating High Performance Pattern Matching on the
  Automata Processor}.
\newblock \bibinfo{journal}{\emph{IEEE Trans. Comput.}} \bibinfo{volume}{68},
  \bibinfo{number}{8} (\bibinfo{year}{2019}), \bibinfo{pages}{1201--1212}.
\newblock
\urldef\tempurl%
\url{https://doi.org/10.1109/TC.2019.2901466}
\showDOI{\tempurl}


\bibitem[\protect\citeauthoryear{RustRegex}{RustRegex}{2021}]%
        {RustRegex}
RustRegex \bibinfo{year}{2021}\natexlab{}.
\newblock \bibinfo{title}{Regex: A Rust Library for Parsing, Compiling, and
  Executing Regular Expressions.}
\newblock \bibinfo{howpublished}{\url{https://github.com/rust-lang/regex}}.
\newblock


\bibitem[\protect\citeauthoryear{Saarikivi, Veanes, Wan, and Xu}{Saarikivi
  et~al\mbox{.}}{2019}]%
        {SaarikiviVWX2019}
\bibfield{author}{\bibinfo{person}{Olli Saarikivi}, \bibinfo{person}{Margus
  Veanes}, \bibinfo{person}{Tiki Wan}, {and} \bibinfo{person}{Eric Xu}.}
  \bibinfo{year}{2019}\natexlab{}.
\newblock \showarticletitle{Symbolic Regex Matcher}. In
  \bibinfo{booktitle}{\emph{Tools and Algorithms for the Construction and
  Analysis of Systems}} \emph{(\bibinfo{series}{LNCS},
  Vol.~\bibinfo{volume}{11427})}. \bibinfo{publisher}{Springer},
  \bibinfo{address}{Cham}, \bibinfo{pages}{372--378}.
\newblock
\urldef\tempurl%
\url{https://doi.org/10.1007/978-3-030-17462-0_24}
\showDOI{\tempurl}


\bibitem[\protect\citeauthoryear{Sadredini, Rahimi, Lenjani, Stan, and
  Skadron}{Sadredini et~al\mbox{.}}{2020}]%
        {SadrediniER20}
\bibfield{author}{\bibinfo{person}{E. Sadredini}, \bibinfo{person}{R. Rahimi},
  \bibinfo{person}{M. Lenjani}, \bibinfo{person}{M. Stan}, {and}
  \bibinfo{person}{K. Skadron}.} \bibinfo{year}{2020}\natexlab{}.
\newblock \showarticletitle{Impala: {{Algorithm}}/{{Architecture
  Co}}-{{Design}} for {{In}}-{{Memory Multi}}-{{Stride Pattern Matching}}}. In
  \bibinfo{booktitle}{\emph{2020 {{IEEE International Symposium}} on {{High
  Performance Computer Architecture}} ({{HPCA}})}}. \bibinfo{publisher}{IEEE},
  \bibinfo{address}{New York, NY, USA}, \bibinfo{pages}{86--98}.
\newblock
\urldef\tempurl%
\url{https://doi.org/10.1109/HPCA47549.2020.00017}
\showDOI{\tempurl}


\bibitem[\protect\citeauthoryear{Sadredini, Rahimi, Verma, Stan, and
  Skadron}{Sadredini et~al\mbox{.}}{2019}]%
        {SRVSS2019eAP}
\bibfield{author}{\bibinfo{person}{Elaheh Sadredini}, \bibinfo{person}{Reza
  Rahimi}, \bibinfo{person}{Vaibhav Verma}, \bibinfo{person}{Mircea Stan},
  {and} \bibinfo{person}{Kevin Skadron}.} \bibinfo{year}{2019}\natexlab{}.
\newblock \showarticletitle{{eAP}: A Scalable and Efficient In-Memory
  Accelerator for Automata Processing}. In
  \bibinfo{booktitle}{\emph{Proceedings of the 52nd Annual IEEE/ACM
  International Symposium on Microarchitecture}} \emph{(\bibinfo{series}{MICRO
  '52})}. \bibinfo{publisher}{ACM}, \bibinfo{address}{New York, NY, USA},
  \bibinfo{pages}{87–--99}.
\newblock
\urldef\tempurl%
\url{https://doi.org/10.1145/3352460.3358324}
\showDOI{\tempurl}


\bibitem[\protect\citeauthoryear{Sidhu and Prasanna}{Sidhu and
  Prasanna}{2001}]%
        {SidhuRP01}
\bibfield{author}{\bibinfo{person}{Reetinder Sidhu} {and}
  \bibinfo{person}{Viktor~K Prasanna}.} \bibinfo{year}{2001}\natexlab{}.
\newblock \showarticletitle{Fast Regular Expression Matching Using FPGAs}. In
  \bibinfo{booktitle}{\emph{The 9th Annual IEEE Symposium on Field-Programmable
  Custom Computing Machines (FCCM '01)}}. \bibinfo{publisher}{IEEE},
  \bibinfo{address}{New York, NY, USA}, \bibinfo{pages}{227--238}.
\newblock


\bibitem[\protect\citeauthoryear{Smith, Estan, Jha, and Kong}{Smith
  et~al\mbox{.}}{2008}]%
        {SmithRE08}
\bibfield{author}{\bibinfo{person}{Randy Smith}, \bibinfo{person}{Cristian
  Estan}, \bibinfo{person}{Somesh Jha}, {and} \bibinfo{person}{Shijin Kong}.}
  \bibinfo{year}{2008}\natexlab{}.
\newblock \showarticletitle{Deflating the Big Bang: Fast and Scalable Deep
  Packet Inspection with Extended Finite Automata}. In
  \bibinfo{booktitle}{\emph{Proceedings of the ACM SIGCOMM 2008 Conference on
  Data Communication}} \emph{(\bibinfo{series}{SIGCOMM '08})}.
  \bibinfo{publisher}{ACM}, \bibinfo{address}{New York, NY, USA},
  \bibinfo{pages}{207–218}.
\newblock
\urldef\tempurl%
\url{https://doi.org/10.1145/1402958.1402983}
\showDOI{\tempurl}


\bibitem[\protect\citeauthoryear{Snort}{Snort}{2022}]%
        {Snort}
Snort \bibinfo{year}{2022}\natexlab{}.
\newblock \bibinfo{title}{Snort Intrusion Detection System.}
\newblock \bibinfo{howpublished}{\url{https://www.snort.org/}}.
\newblock


\bibitem[\protect\citeauthoryear{Sourdis, Bispo, Cardoso, and
  Vassiliadis}{Sourdis et~al\mbox{.}}{2008}]%
        {SourdisIB08}
\bibfield{author}{\bibinfo{person}{Ioannis Sourdis}, \bibinfo{person}{Joao
  Bispo}, \bibinfo{person}{Joao~MP Cardoso}, {and} \bibinfo{person}{Stamatis
  Vassiliadis}.} \bibinfo{year}{2008}\natexlab{}.
\newblock \showarticletitle{Regular Expression Matching in Reconfigurable
  Hardware}.
\newblock \bibinfo{journal}{\emph{Journal of Signal Processing Systems}}
  \bibinfo{volume}{51}, \bibinfo{number}{1} (\bibinfo{year}{2008}),
  \bibinfo{pages}{99--121}.
\newblock
\urldef\tempurl%
\url{https://doi.org/10.1007/s11265-007-0131-0}
\showDOI{\tempurl}


\bibitem[\protect\citeauthoryear{SPICE}{SPICE}{2022}]%
        {SPICE}
SPICE \bibinfo{year}{2022}\natexlab{}.
\newblock \bibinfo{title}{SPICE: A General-purpose Circuit Simulation Program
  for Nonlinear DC, Nonlinear Transient, and Linear AC Analyses}.
\newblock
  \bibinfo{howpublished}{\url{http://bwrcs.eecs.berkeley.edu/Classes/IcBook/SPICE}}.
\newblock


\bibitem[\protect\citeauthoryear{Stockmeyer and Meyer}{Stockmeyer and
  Meyer}{1973}]%
        {StockmeyerM1973}
\bibfield{author}{\bibinfo{person}{Larry~J. Stockmeyer} {and}
  \bibinfo{person}{Albert~R. Meyer}.} \bibinfo{year}{1973}\natexlab{}.
\newblock \showarticletitle{Word Problems Requiring Exponential Time
  (Preliminary Report)}. In \bibinfo{booktitle}{\emph{Proceedings of the Fifth
  Annual ACM Symposium on Theory of Computing}} \emph{(\bibinfo{series}{STOC
  '73})}. \bibinfo{publisher}{ACM}, \bibinfo{address}{New York, NY, USA},
  \bibinfo{pages}{1--9}.
\newblock
\urldef\tempurl%
\url{https://doi.org/10.1145/800125.804029}
\showDOI{\tempurl}


\bibitem[\protect\citeauthoryear{Subramaniyan, Wang, Balasubramanian, Blaauw,
  Sylvester, and Das}{Subramaniyan et~al\mbox{.}}{2017}]%
        {SubramaniyanAW17}
\bibfield{author}{\bibinfo{person}{Arun Subramaniyan},
  \bibinfo{person}{Jingcheng Wang}, \bibinfo{person}{Ezhil R.~M.
  Balasubramanian}, \bibinfo{person}{David Blaauw}, \bibinfo{person}{Dennis
  Sylvester}, {and} \bibinfo{person}{Reetuparna Das}.}
  \bibinfo{year}{2017}\natexlab{}.
\newblock \showarticletitle{Cache Automaton}. In
  \bibinfo{booktitle}{\emph{Proceedings of the 50th Annual IEEE/ACM
  International Symposium on Microarchitecture}}
  \emph{(\bibinfo{series}{MICRO-50 '17})}. \bibinfo{publisher}{ACM},
  \bibinfo{address}{New York, NY, USA}, \bibinfo{pages}{259–272}.
\newblock
\urldef\tempurl%
\url{https://doi.org/10.1145/3123939.3123986}
\showDOI{\tempurl}


\bibitem[\protect\citeauthoryear{Suricata}{Suricata}{2022}]%
        {Suricata}
Suricata \bibinfo{year}{2022}\natexlab{}.
\newblock \bibinfo{title}{Suricata Threat Detection Engine.}
\newblock \bibinfo{howpublished}{\url{https://suricata.io/}}.
\newblock


\bibitem[\protect\citeauthoryear{Tandon, Sleiman, Cafarella, and
  Wenisch}{Tandon et~al\mbox{.}}{2016}]%
        {tandonPS16}
\bibfield{author}{\bibinfo{person}{Prateek Tandon}, \bibinfo{person}{Faissal~M
  Sleiman}, \bibinfo{person}{Michael~J Cafarella}, {and}
  \bibinfo{person}{Thomas~F Wenisch}.} \bibinfo{year}{2016}\natexlab{}.
\newblock \showarticletitle{Hawk: Hardware Support for Unstructured Log
  Processing}. In \bibinfo{booktitle}{\emph{2016 IEEE 32nd International
  Conference on Data Engineering (ICDE)}}. \bibinfo{publisher}{IEEE},
  \bibinfo{address}{New York, NY, USA}, \bibinfo{pages}{469--480}.
\newblock
\urldef\tempurl%
\url{https://doi.org/10.1109/ICDE.2016.7498263}
\showDOI{\tempurl}


\bibitem[\protect\citeauthoryear{Thompson}{Thompson}{1968}]%
        {Thompson1968}
\bibfield{author}{\bibinfo{person}{Ken Thompson}.}
  \bibinfo{year}{1968}\natexlab{}.
\newblock \showarticletitle{Programming Techniques: Regular Expression Search
  Algorithm}.
\newblock \bibinfo{journal}{\emph{Commun. ACM}} \bibinfo{volume}{11},
  \bibinfo{number}{6} (\bibinfo{year}{1968}), \bibinfo{pages}{419--422}.
\newblock
\urldef\tempurl%
\url{https://doi.org/10.1145/363347.363387}
\showDOI{\tempurl}


\bibitem[\protect\citeauthoryear{Tuck, Sherwood, Calder, and Varghese}{Tuck
  et~al\mbox{.}}{2004}]%
        {TuckNS04}
\bibfield{author}{\bibinfo{person}{N. Tuck}, \bibinfo{person}{T. Sherwood},
  \bibinfo{person}{B. Calder}, {and} \bibinfo{person}{G. Varghese}.}
  \bibinfo{year}{2004}\natexlab{}.
\newblock \showarticletitle{Deterministic Memory-efficient String Matching
  Algorithms for Intrusion Detection}. In \bibinfo{booktitle}{\emph{IEEE
  INFOCOM 2004}}, Vol.~\bibinfo{volume}{4}. \bibinfo{publisher}{IEEE},
  \bibinfo{address}{New York, NY, USA}, \bibinfo{pages}{2628--2639 vol.4}.
\newblock
\urldef\tempurl%
\url{https://doi.org/10.1109/INFCOM.2004.1354682}
\showDOI{\tempurl}


\bibitem[\protect\citeauthoryear{Turo\v{n}ov\'{a}, Hol\'{\i}k, Leng\'{a}l,
  Saarikivi, Veanes, and Vojnar}{Turo\v{n}ov\'{a} et~al\mbox{.}}{2020}]%
        {TuronovaHLSVV2020}
\bibfield{author}{\bibinfo{person}{Lenka Turo\v{n}ov\'{a}},
  \bibinfo{person}{Luk\'{a}\v{s} Hol\'{\i}k}, \bibinfo{person}{Ond\v{r}ej
  Leng\'{a}l}, \bibinfo{person}{Olli Saarikivi}, \bibinfo{person}{Margus
  Veanes}, {and} \bibinfo{person}{Tom\'{a}\v{s} Vojnar}.}
  \bibinfo{year}{2020}\natexlab{}.
\newblock \showarticletitle{Regex Matching with Counting-Set Automata}.
\newblock \bibinfo{journal}{\emph{Proceedings of the ACM on Programming
  Languages}} \bibinfo{volume}{4}, \bibinfo{number}{OOPSLA}, Article
  \bibinfo{articleno}{218} (\bibinfo{year}{2020}),
  \bibinfo{numpages}{30}~pages.
\newblock
\urldef\tempurl%
\url{https://doi.org/10.1145/3428286}
\showDOI{\tempurl}


\bibitem[\protect\citeauthoryear{Vasiliadis, Polychronakis, Antonatos,
  Markatos, and Ioannidis}{Vasiliadis et~al\mbox{.}}{2009}]%
        {VasiliadisGP09}
\bibfield{author}{\bibinfo{person}{Giorgos Vasiliadis},
  \bibinfo{person}{Michalis Polychronakis}, \bibinfo{person}{Spiros Antonatos},
  \bibinfo{person}{Evangelos~P. Markatos}, {and} \bibinfo{person}{Sotiris
  Ioannidis}.} \bibinfo{year}{2009}\natexlab{}.
\newblock \showarticletitle{Regular Expression Matching on Graphics Hardware
  for Intrusion Detection}. In \bibinfo{booktitle}{\emph{Recent Advances in
  Intrusion Detection}}, \bibfield{editor}{\bibinfo{person}{Engin Kirda},
  \bibinfo{person}{Somesh Jha}, {and} \bibinfo{person}{Davide Balzarotti}}
  (Eds.). \bibinfo{publisher}{Springer}, \bibinfo{address}{Heidelberg},
  \bibinfo{pages}{265--283}.
\newblock
\urldef\tempurl%
\url{https://doi.org/10.1007/978-3-642-04342-0_14}
\showDOI{\tempurl}


\bibitem[\protect\citeauthoryear{Wadden, Dang, Brunelle, Tracy~II, Guo,
  Sadredini, Wang, Bo, Robins, Stan, and Skadron}{Wadden et~al\mbox{.}}{2016}]%
        {ANMLZoo}
\bibfield{author}{\bibinfo{person}{Jack Wadden}, \bibinfo{person}{Vinh Dang},
  \bibinfo{person}{Nathan Brunelle}, \bibinfo{person}{Tommy Tracy~II},
  \bibinfo{person}{Deyuan Guo}, \bibinfo{person}{Elaheh Sadredini},
  \bibinfo{person}{Ke Wang}, \bibinfo{person}{Chunkun Bo},
  \bibinfo{person}{Gabriel Robins}, \bibinfo{person}{Mircea Stan}, {and}
  \bibinfo{person}{Kevin Skadron}.} \bibinfo{year}{2016}\natexlab{}.
\newblock \showarticletitle{{ANMLZoo}: A Benchmark Suite for Exploring
  Bottlenecks in Automata Processing Engines and Architectures}. In
  \bibinfo{booktitle}{\emph{2016 IEEE International Symposium on Workload
  Characterization (IISWC)}}. \bibinfo{publisher}{IEEE},
  \bibinfo{pages}{1--12}.
\newblock
\urldef\tempurl%
\url{https://doi.org/10.1109/IISWC.2016.7581271}
\showDOI{\tempurl}


\bibitem[\protect\citeauthoryear{Wadden, Tracy, Sadredini, Wu, Bo, Du, Wei,
  Udall, Wallace, Stan, and Skadron}{Wadden et~al\mbox{.}}{2018}]%
        {WaddenJT18}
\bibfield{author}{\bibinfo{person}{Jack Wadden}, \bibinfo{person}{Tommy Tracy},
  \bibinfo{person}{Elaheh Sadredini}, \bibinfo{person}{Lingxi Wu},
  \bibinfo{person}{Chunkun Bo}, \bibinfo{person}{Jesse Du},
  \bibinfo{person}{Yizhou Wei}, \bibinfo{person}{Jeffrey Udall},
  \bibinfo{person}{Matthew Wallace}, \bibinfo{person}{Mircea Stan}, {and}
  \bibinfo{person}{Kevin Skadron}.} \bibinfo{year}{2018}\natexlab{}.
\newblock \showarticletitle{AutomataZoo: A Modern Automata Processing Benchmark
  Suite}. In \bibinfo{booktitle}{\emph{2018 IEEE International Symposium on
  Workload Characterization (IISWC)}}. \bibinfo{publisher}{IEEE},
  \bibinfo{address}{New York, NY, USA}, \bibinfo{pages}{13--24}.
\newblock
\urldef\tempurl%
\url{https://doi.org/10.1109/IISWC.2018.8573482}
\showDOI{\tempurl}


\bibitem[\protect\citeauthoryear{Wang, Pu, Knezek, and Liu}{Wang
  et~al\mbox{.}}{2010}]%
        {WangHP10}
\bibfield{author}{\bibinfo{person}{Hao Wang}, \bibinfo{person}{Shi Pu},
  \bibinfo{person}{Gabriel Knezek}, {and} \bibinfo{person}{Jyh-Charn Liu}.}
  \bibinfo{year}{2010}\natexlab{}.
\newblock \showarticletitle{A Modular NFA Architecture for Regular Expression
  Matching}. In \bibinfo{booktitle}{\emph{Proceedings of the 18th Annual
  ACM/SIGDA International Symposium on Field Programmable Gate Arrays}}
  \emph{(\bibinfo{series}{FPGA '10})}. \bibinfo{publisher}{ACM},
  \bibinfo{address}{New York, NY, USA}, \bibinfo{pages}{209–218}.
\newblock
\urldef\tempurl%
\url{https://doi.org/10.1145/1723112.1723149}
\showDOI{\tempurl}


\bibitem[\protect\citeauthoryear{Wang, Angstadt, Bo, Brunelle, Sadredini,
  Tracy, Wadden, Stan, and Skadron}{Wang et~al\mbox{.}}{2016}]%
        {WangKA16}
\bibfield{author}{\bibinfo{person}{Ke Wang}, \bibinfo{person}{Kevin Angstadt},
  \bibinfo{person}{Chunkun Bo}, \bibinfo{person}{Nathan Brunelle},
  \bibinfo{person}{Elaheh Sadredini}, \bibinfo{person}{Tommy Tracy},
  \bibinfo{person}{Jack Wadden}, \bibinfo{person}{Mircea Stan}, {and}
  \bibinfo{person}{Kevin Skadron}.} \bibinfo{year}{2016}\natexlab{}.
\newblock \showarticletitle{An Overview of Micron's Automata Processor}. In
  \bibinfo{booktitle}{\emph{Proceedings of the Eleventh IEEE/ACM/IFIP
  International Conference on Hardware/Software Codesign and System Synthesis}}
  \emph{(\bibinfo{series}{CODES '16})}. \bibinfo{publisher}{ACM},
  \bibinfo{address}{New York, NY, USA}, Article \bibinfo{articleno}{14},
  \bibinfo{numpages}{3}~pages.
\newblock
\urldef\tempurl%
\url{https://doi.org/10.1145/2968456.2976763}
\showDOI{\tempurl}


\bibitem[\protect\citeauthoryear{Wang, Hong, Chang, Park, Langdale, Hu, and
  Zhu}{Wang et~al\mbox{.}}{2019}]%
        {WangHCPLHZ2019Hyperscan}
\bibfield{author}{\bibinfo{person}{Xiang Wang}, \bibinfo{person}{Yang Hong},
  \bibinfo{person}{Harry Chang}, \bibinfo{person}{KyoungSoo Park},
  \bibinfo{person}{Geoff Langdale}, \bibinfo{person}{Jiayu Hu}, {and}
  \bibinfo{person}{Heqing Zhu}.} \bibinfo{year}{2019}\natexlab{}.
\newblock \showarticletitle{Hyperscan: A Fast Multi-Pattern Regex Matcher for
  Modern {CPUs}}. In \bibinfo{booktitle}{\emph{16th USENIX Symposium on
  Networked Systems Design and Implementation (NSDI '19)}}.
  \bibinfo{publisher}{USENIX Association}, \bibinfo{pages}{631--648}.
\newblock
\urldef\tempurl%
\url{https://www.usenix.org/conference/nsdi19/presentation/wang-xiang}
\showURL{%
\tempurl}


\bibitem[\protect\citeauthoryear{Xie, Dang, Wadden, Skadron, and Stan}{Xie
  et~al\mbox{.}}{2017}]%
        {XieTV17}
\bibfield{author}{\bibinfo{person}{Ted Xie}, \bibinfo{person}{Vinh Dang},
  \bibinfo{person}{Jack Wadden}, \bibinfo{person}{Kevin Skadron}, {and}
  \bibinfo{person}{Mircea Stan}.} \bibinfo{year}{2017}\natexlab{}.
\newblock \showarticletitle{REAPR: Reconfigurable Engine for Automata
  Processing}. In \bibinfo{booktitle}{\emph{2017 27th International Conference
  on Field Programmable Logic and Applications (FPL)}}.
  \bibinfo{publisher}{IEEE}, \bibinfo{address}{New York, NY, USA},
  \bibinfo{pages}{1--8}.
\newblock
\urldef\tempurl%
\url{https://doi.org/10.23919/FPL.2017.8056759}
\showDOI{\tempurl}


\bibitem[\protect\citeauthoryear{Yang, Jiang, and Prasanna}{Yang
  et~al\mbox{.}}{2008}]%
        {YangYJ08}
\bibfield{author}{\bibinfo{person}{Yi-Hua~E. Yang}, \bibinfo{person}{Weirong
  Jiang}, {and} \bibinfo{person}{Viktor~K. Prasanna}.}
  \bibinfo{year}{2008}\natexlab{}.
\newblock \showarticletitle{Compact Architecture for High-Throughput Regular
  Expression Matching on FPGA}. In \bibinfo{booktitle}{\emph{Proceedings of the
  4th ACM/IEEE Symposium on Architectures for Networking and Communications
  Systems}} \emph{(\bibinfo{series}{ANCS '08})}. \bibinfo{publisher}{ACM},
  \bibinfo{address}{New York, NY, USA}, \bibinfo{pages}{30–39}.
\newblock
\urldef\tempurl%
\url{https://doi.org/10.1145/1477942.1477948}
\showDOI{\tempurl}


\bibitem[\protect\citeauthoryear{Yang and Prasanna}{Yang and Prasanna}{2011}]%
        {YangYP2011}
\bibfield{author}{\bibinfo{person}{Yi-Hua~E. Yang} {and}
  \bibinfo{person}{Viktor~K. Prasanna}.} \bibinfo{year}{2011}\natexlab{}.
\newblock \showarticletitle{Space-time Tradeoff in Regular Expression Matching
  with Semi-deterministic Finite Automata}. In \bibinfo{booktitle}{\emph{2011
  Proceedings IEEE INFOCOM}}. \bibinfo{publisher}{IEEE}, \bibinfo{address}{New
  York, NY, USA}, \bibinfo{pages}{1853--1861}.
\newblock
\urldef\tempurl%
\url{https://doi.org/10.1109/INFCOM.2011.5934986}
\showDOI{\tempurl}


\bibitem[\protect\citeauthoryear{Yu, Chen, Diao, Lakshman, and Katz}{Yu
  et~al\mbox{.}}{2006}]%
        {YuCDLK2006}
\bibfield{author}{\bibinfo{person}{Fang Yu}, \bibinfo{person}{Zhifeng Chen},
  \bibinfo{person}{Yanlei Diao}, \bibinfo{person}{T.~V. Lakshman}, {and}
  \bibinfo{person}{Randy~H. Katz}.} \bibinfo{year}{2006}\natexlab{}.
\newblock \showarticletitle{Fast and Memory-Efficient Regular Expression
  Matching for Deep Packet Inspection}. In
  \bibinfo{booktitle}{\emph{Proceedings of the 2006 ACM/IEEE Symposium on
  Architecture for Networking and Communications Systems}}
  \emph{(\bibinfo{series}{ANCS '06})}. \bibinfo{publisher}{ACM},
  \bibinfo{address}{New York, NY, USA}, \bibinfo{pages}{93--102}.
\newblock
\urldef\tempurl%
\url{https://doi.org/10.1145/1185347.1185360}
\showDOI{\tempurl}


\bibitem[\protect\citeauthoryear{Zu, Yang, Xu, Wang, Tian, Peng, and Dong}{Zu
  et~al\mbox{.}}{2012}]%
        {ZuYY12}
\bibfield{author}{\bibinfo{person}{Yuan Zu}, \bibinfo{person}{Ming Yang},
  \bibinfo{person}{Zhonghu Xu}, \bibinfo{person}{Lin Wang},
  \bibinfo{person}{Xin Tian}, \bibinfo{person}{Kunyang Peng}, {and}
  \bibinfo{person}{Qunfeng Dong}.} \bibinfo{year}{2012}\natexlab{}.
\newblock \showarticletitle{GPU-Based NFA Implementation for Memory Efficient
  High Speed Regular Expression Matching}. In
  \bibinfo{booktitle}{\emph{Proceedings of the 17th ACM SIGPLAN Symposium on
  Principles and Practice of Parallel Programming}}
  \emph{(\bibinfo{series}{PPoPP '12})}. \bibinfo{publisher}{ACM},
  \bibinfo{address}{New York, NY, USA}, \bibinfo{pages}{129–140}.
\newblock
\urldef\tempurl%
\url{https://doi.org/10.1145/2145816.2145833}
\showDOI{\tempurl}


\end{thebibliography}
\end{document}